\documentclass[reprint,aps,showpacs,pra,nofootinbib]{revtex4-2}
\usepackage{color}
\usepackage{graphicx}
\usepackage{amssymb}
\usepackage{amsmath}
\usepackage{amsthm}
\usepackage{verbatim}
\usepackage{mathtools}
\usepackage[latin1]{inputenc}
\usepackage{bbm}
\usepackage{multirow}
\usepackage[thinlines]{easytable}
\usepackage{braket}
\usepackage[varg]{txfonts}
\usepackage{hyperref}
\usepackage{array}
\usepackage{bbding}
\usepackage{pifont}
\usepackage{wasysym}

\hypersetup{colorlinks=true,linkcolor=black,urlcolor=blue,citecolor=blue}

\newcommand{\cmark}{\ding{52}}
\newcommand{\xmark}{\ding{56}}
\newcommand{\be}{\begin{equation}}
\newcommand{\ee}{\end{equation}}
\newcommand{\beq}{\begin{eqnarray}}
\newcommand{\eeq}{\end{eqnarray}}
\newcommand{\Tr}{\mathrm{Tr}\,}
\newcommand{\norm}[2]{\mbox{$|\!| #1 |\!|_{#2}$}}
\newcommand{\bignorm}[2]{\mbox{$\left|\! \left| #1 \right|\!\right|_{#2}$}}

\newcommand{\tD}{\widetilde{D}}

\newcommand{\tfR}{\widetilde{\mathfrak{R}}}
\newcommand{\fR}{\mathfrak{R}}
\newcommand{\fI}{\mathfrak{I}}
\newcommand{\up}{\uparrow}
\newcommand{\down}{\downarrow}
\newcommand{\bigdbar}{\,\Big|\Big|\,}
\newcommand{\mc}{\mathcal}
\newcommand{\mf}{\mathfrak}
\newcommand{\mbb}{\mathbbm}

\newtheorem{theorem}{Theorem}

\newtheorem{lemma}{Lemma}

\newtheorem{axiom}{Axiom}

\theoremstyle{definition}
\newtheorem{definition}{Definition}
\newtheorem{remark}{Remark}

\theoremstyle{remark}
\newtheorem{example}{Example}

\begin{document}

\title{Quantum realism: axiomatization and quantification}

\author{Alexandre C. Orthey Jr.}
\author{R. M. Angelo}
\affiliation{Department of Physics, Federal University of Paran\'a, P.O. Box 19044, 81531-980 Curitiba, Paran\'a, Brazil}

\begin{abstract}
The emergence of an objective reality in line with the laws of the microscopic world has been the focus of longstanding debates. Recent approaches seem to have reached a consensus at least with respect to one aspect, namely, that the encoding of information about a given observable in a physical degree freedom is a necessary condition for such an observable to become an element of the physical reality. Taking this as a fundamental premise and inspired by quantum information theory, here we build an axiomatization for  quantum realism---a notion of realism compatible with quantum theory. Our strategy consists of listing some physically motivated principles able to characterize quantum realism in a ``metric'' independent manner. We introduce some criteria defining monotones and measures of realism and then search for potential candidates within some celebrated information theories---those induced by the von Neumann, R\'enyi, and Tsallis entropies. We explicitly construct some classes of entropic quantifiers, among which some are shown to satisfy all the proposed axioms and hence can be taken as faithful estimates for the degree of reality (or definiteness) of a given physical observable. Hopefully, our framework may offer a formal ground for further discussions on  foundational aspects of quantum mechanics.
\end{abstract}

\maketitle

\section{Introduction}

The explanation of quantum phenomena in terms of an objective pre-existent reality is arguably problematic. Unless one is willing to accept nonlocal elements of reality---such as the Bohmian trajectories---it seems better to abandon the idea that an electron always chooses to traverse only one of two slits or to travel a well defined orbit around a proton. Challenging our everyday intuition, quantum mechanics allows a ``quantum coin'' to be prepared in a state like $\ket{\psi}=\sqrt{1-p}\ket{\mathrm{H}}+e^{i\phi}\sqrt{p}\ket{\mathrm{T}}$, with $\{\mathrm{H},\mathrm{T}\}\equiv\{\mathrm{Heads},\mathrm{Tails}\}$ and $p\in(0,1)$. This preparation can by no means be described by an ensemble $(1-p)\ket{\mathrm{H}}\bra{\mathrm{H}}+p\ket{\mathrm{T}}\bra{\mathrm{T}}$ of ``well behaved'' coins, with a fraction $p$ of copies with tails facing upward. In particular, this classical mixture is incapable of encapsulating the fundamental phase $\phi$. Therefore, the idea that a two-valued physical quantity, like an electron's spin or a photon's polarization, is an element of reality---thus being well defined regardless of any observation \textcolor{blue}{(realism)}---does not peacefully coexist with the preparation $\ket{\psi}$.

The discussion about elements of the physical reality in the context of quantum mechanics takes us back to the seminal work of Einstein, Podolsky, and Rosen (EPR) \cite{EPR1935}, where the authors call into question the completeness of quantum theory. Envisaging a scenario in which measurements of incompatible observables are conducted in two spatially separated parts of an entangled system and dismissing any sort of action at distance, EPR concluded that there exist elements of physical reality that are not predicted by  quantum theory, which is then alleged to be an incomplete model of nature. This idea was immediately confronted by Bohr \cite{bohr1935}, who argued that complementary physical quantities associated with incompatible observables cannot be elements of reality in the same experimental arrangement. Decades later, Bell (first in Ref. \cite{bell1964} and then in Ref. \cite{bell1976}) proved that any model based on local hidden variables cannot be consistent with the predictions of quantum mechanics. Given the undeniable success of quantum mechanics in fitting experimental data, violations of Bell inequalities suggest that nature itself is incompatible with the local causality hypothesis \cite{bell1976}. This phenomenon, conventionally referred to as Bell nonlocality \cite{brunner2014}, has been verified through several loophole-free tests \cite{hensen2015,giustina2015,shalm2015,hensen2016,rauch2018,li2018}. Interestingly, local causality has been acknowledged as a compound assumption \cite{zukowski2014}, stronger than locality (in the sense of\ ``parameter independence'' \cite{shimony1984}) but weaker than no-signaling (the impossibility of faster-than-light communication, which has never been seen violated), so that no tension whatsoever exists with relativity principles. The debate still remains concerning an alternative decomposition of the local causality hypothesis into other assumptions, such as some form of realism (see the work of Wiseman and Cavalcanti \cite{wiseman2017,cavalcanti2021} for a detailed discussion on the assumptions underlying Bell's theorem).

Recently, the emergence of physical reality has been discussed in scenarios involving more than one observer. Considering extended Wigner's friend scenarios, Brukner derived a no-go theorem for observer-independent facts \cite{brukner2018}. In another no-go theorem, Frauchiger and Renner showed that quantum mechanics cannot consistently describe the use of itself \cite{frauchiger2018}. Inspired by that, Bong {\it et al.} \cite{bong2020} have proved and experimentally verified that if quantum evolution is controllable on the scale of an observer, then no physical theory can simultaneously satisfy the hypotheses of no-superdeterminism, locality, and absoluteness of observed events (also called macroreality \cite{wiseman2017}). Roughly speaking, by examining instances where two observers confront their experiences about the physical reality, these results reinforce the subtleties underlying the measurement problem. In a different vein, the authors of Ref. \cite{dieguez2018} suggest that the elements of reality associated with the system under scrutiny are established when correlations are developed in an early stage of the dynamics, before any observer comes into play.  

Whatever perspective one may adopt in assessing the quantum phenomena, the task of combining the algebraic structure of the theory with the experienced physical reality is always an issue. In effect, it has been suggested that many of the interpretations of quantum mechanics known to date can be divided in two groups, depending on their attitude toward (the emergence of) realism\footnote{In most cases the term realism is taken as a synonym for ``classical reality'', which may be identified with the dogma according to which all the systems exist and have well defined physical properties at every instant of time regardless the presence or action of observers (brain-endowed systems).} \cite{cabello2017,fuchs2017}. Amongst the frameworks accounting for the emergence of an objective reality from the quantum substratum, quantum Darwinism \cite{zurek2009} is a prominent one. Corroborated by recent experiments \cite{ciampini2018,chen2019,unden2019}, this model claims that reality emerges when information about a quantum system gets prolifically copied into the environment.

Once we accept a noninstantaneous transition to classical reality, it makes sense thinking of an intermediate state of affairs, the one prior to the definitive achievement of realism. Presumably, any gradation of ``nonrealism'' would be possible {\it a priori}. This was the intuition leading Bilobran and Angelo (BA) to introduce the so-called {\it irreality} (the complement of reality)---an operational quantifier intended to diagnose how far a given physical quantity is from full definiteness \cite{bilobran2015}. The criterion of realism envisaged in this approach, henceforth referred to as BA's realism, does not imply full classical reality, since situations are shown to exist where the $z$-component of spin is an element of reality whereas the $x$-component is not. The basic premise employed by BA is that a measurement establishes an element of reality for the measured observable, even when the measurement outcome is not revealed. If a given state is not altered by an unrevealed measurement, then this means that this state already implied an element of reality before the measurement. Hence, the uncertainties associated with the measured observable are of subjective essence and the state is epistemic. Many developments followed from this framework, from a novel notion of nonlocality \cite{gomes2018,orthey2019,fucci2019} to foundational aspects of quantum theory \cite{dieguez2018,engelbert2020} and their proofs of principle \cite{mancino2018,dieguez2021}, to the realization that irreality is a quantum resource \cite{costa2020}.

Given the above, it seems very difficult to figure out what quantum mechanics is all about without a proper framing of the notion of realism. This work is devoted to this task. Inspired by the formal structure of the resource theories of entanglement \cite{plenio2014} and coherence \cite{streltsov2017} and resorting to a fundamental premise shared by quantum Darwinism and BA's realism, we suggest an axiomatization of the concept of {\it quantum realism} (in opposition to classical reality). Our axioms are physically motivated and connected to an informational description of the measurement dynamics. We propose two categories of reality quantifiers: reality monotones and reality measures, the former requiring a smaller set of axioms to be satisfied. Our search for quantifiers takes place within the quantum information theories of R\'enyi~\cite{petz1986,renner2005,datta2009,tomamichel2009,mosonyi2011,frank2013,beigi2013,lennert2013,wilde2014,datta2014,tomamichel2014,gupta2015,mosonyi2015,tomamichel2015,mckinlay2020} and Tsallis~\cite{tsallis1988,raggio1995,abe2003,hiai2011,rajagopal2014,rastegin2016}, whose scopes extend the one induced by the von Neumann entropy~\cite{nielsen2000}. This article is organized as follows. In Sec.~\ref{secAxioms}, we present our list of axioms for quantum realism. In Sec.~\ref{secElements}, we briefly review some elements of the aforementioned quantum information theories. In Sec.~\ref{secMonotones}, we explicitly build reality monotones and measures in consonance with the proposed axioms. Our concluding remarks are left to Sec.~\ref{secDiscussion}.

Before starting, it is useful to spell out the meaning that shall be presumed from the term ``quantum realism'' throughout this article. It is not connected to the {\it existence} of a system, which is taken for granted from the outset, but rather to the {\it definiteness of a physical quantity} prior to any observer's intervention. When this scenario is realized, the corresponding quantity is said to be an {\it element of reality}. Unlike classical reality, quantum realism does not presume all physical quantities to be elements of reality simultaneously. ``Definiteness'', by its turn, does not mean ``total absence of uncertainties'', which would be equivalent to the condition of full predictability appearing in EPR's approach. It actually refers to the absence of quantum uncertainties for a particular physical quantity. Adhering to BA's criterion, the quantum state $\rho=(1-p)\ket{a_1}\bra{a_1}+p\ket{a_2}\bra{a_2}$, with $A\ket{a_{1,2}}=a_{1,2}\ket{a_{1,2}}$, is then taken as an example of scenario in which quantum realism is established for the observable $A$, even though mere subjective uncertainties are present when $p\in(0,1)$. To see this, we check what happens with the state when it is submitted to a nonselective measurement of $A=\sum_i a_i A_i$ (where $A_i=\ket{a_i}\bra{a_i}$). Since $\sum_iA_i\rho A_i=\rho$, then BA's criterion of realism is satisfied, meaning that the element of reality that would presumably be installed by the measurement is already there before the measurement. In other words, in this case, reality is not dictated by the measurement and the present uncertainties reflect only subjective ignorance ($\rho$ is an epistemic state). This is consistent with the fact that no ``interference pattern'' would be observed with respect to the observable $A$. Thus, while in BA's approach $A$ is an element of reality for all $p$, in EPR's this is so only when $p=0$ or $p=1$ (full predictability regimes). The differences between these approaches can be further emphasized in multipartite settings. As thoroughly discussed in Ref.~\cite{fucci2019}, while EPR would claim that the spin observables $S_{x,z}$ are simultaneous elements of reality for the singlet state, BA's criterion implies that these observables actually are maximally unreal. As another example, consider the bipartite separable state $\rho_\text{sep}=\sum_{\lambda}p_\lambda\rho^\mc{A}_\lambda\otimes\rho^\mc{B}_\lambda$. It immediately satisfies Bell's local causality hypothesis---sometimes referred to as local realism---but does not imply BA realism for a vast set of observables, since $\sum_i(A_i\otimes\mbb{1}_\mc{B})\rho_\text{sep}(A_i\otimes\mbb{1}_\mc{B})\neq \rho_\text{sep}$. Throughout this article, we stick to the notion of BA realism.

\section{Axioms for quantum realism}\label{secAxioms}

Quantum resource theories have shown to be a powerful framework to characterize a given quantum effect \cite{chitambar2019}. Incidentally, quantum realism cannot be thought of as a quantum resource because reality abounds for free in the classical regime. On the other hand, quantum realism is complementary to quantum irrealism (as quantified by irreality, which is believed to be a quantum resource \cite{costa2020}). With this inversion in mind, we seek inspiration in the formal structure of quantum resource theories to guide our axiomatization of quantum realism.

We start by grounding our intuition on some empirical facts. After passing through a wall with two slits, an electron has its paths described as a quantum superposition and an interference pattern is observed in the detection system. During the flight, quantum mechanics does not ascribe a well defined position for the electron, so that its position is not an element of reality and the electron is said to behave like a wave. On the other hand, when the two slits are preceded with a very lightweight floating slit, the interference pattern disappears \cite{bohr1935,liu2015} (see the double-slit quantum eraser \cite{walborn2002} for a similar phenomenology). In this case, the entanglement created between the electron and the floating slit allows for the former to be described by a statistical mixture. It then follows that trajectory-based models are admissible so that the electron position can be claimed to be an element of reality. In other words, particle-like elements of reality emerge in this experiment because a given degree of freedom---the momentum of the lightweight slit---encodes which-way information about the electron~\cite{angelo2015}. Notice that these are expected to be the results of the experiment even in the absence of a huge environment, like a thermal bath.

Now, even though the supporters of quantum Darwinism would eventually claim that the conditions for the emergence of an objective reality are not met during the electron flight---for the information about the electron path does not have an environment to be recorded in---we believe they would agree that the motional degree of freedom of the first slit is able to acquire information about the electron path, thus suppressing its wave-like properties. This is exactly the same viewpoint adopted by supporters of BA's realism \cite{dieguez2018,mancino2018}. We then take this common  perspective as our fundamental premise regarding the dynamical emergence of quantum realism: \textit{the reality status of a physical observable can increase only when information about it is stored in another physical degree of freedom}. Moreover, we adhere to BA's conception that such reality degree can be quantified at every instant of time by use of the quantum state.

To formalize these ideas, we consider the functional $\rho\mapsto\fR_A(\rho)$, hereafter named {\it the reality of the observable $A\in\mathfrak{B}(\mathcal{H_A})$ given the state} $\rho\in\mathfrak{B}(\mathcal{H_S=H_A\otimes H_B})$, where $\mathfrak{B}(\mathcal{H})$ is the set of positive semidefinite Hermitian operators acting on the Hilbert space $\mathcal{H}$. Let us now consider some generic dynamics involving the interest system $\mc{S}$ and an ancillary system $\mc{E}$ generically referred to as environment. Assume that the state of the composite system at an arbitrary instant of time $t$ is given by $\upsilon_t\in\mf{B}(\mc{H_S\otimes H_E})$, so that $\rho_t=\Tr_\mc{E}(\upsilon_t)$ denotes the reduced state of $\mc{S}$ with initial condition $\rho_{t=0}=\rho$. An alteration in the reality degree of the observable $A$ in the time interval $[t_1,t_2]$ is here denoted as
\be\label{delta_R}
\Delta\fR_A(t_2,t_1)\coloneqq \fR_A(\rho_{t_2})-\fR_A(\rho_{t_1}).
\ee
Let us also introduce another functional, $\upsilon_t\mapsto I_\mc{E|S}(\upsilon_t)$, aiming at denoting how certain informational content associated with the environment is conditioned to some configuration of the system. Variations of this information with time are then described as
\be\label{delta_I}
\Delta I_{\mc{E|S}}(t_2,t_1)\coloneqq I_{\mc{E|S}}(\upsilon_{t_2})-I_{\mc{E|S}}(\upsilon_{t_1}).
\ee
We are now ready to state our main postulate.

\begin{axiom}[Reality and information flow]\label{A_reality_inf}
The degree of reality of an observable $A$ is altered in the time interval $[t_1,t_2]$ only when an amount $\Delta I_\mc{E|S}(t_2,t_1)$ of information about this observable is shared with the environment, that is,
\be\label{DR=DI}
\Delta\fR_A(t_2,t_1)\equiv \Delta I_\mc{E|S}(t_2,t_1).
\ee
\end{axiom}
\noindent The specific mathematical structures of $\fR_A$ and $I_\mc{E|S}$ and the sense in which information ``flows'' to the environment will be opportunely specified for each information theory we consider. By now, the crux is realizing that this axiom implements the fundamental premise of quantum Darwinism and BA's realism, namely, that reality varies with time only through a physical process involving interactions, the establishment of correlations, and some form of information exchange. Also, the relation \eqref{DR=DI} attaches an informational profile to the quantifier $\fR_A$. Although this choice is somewhat {\it ad hoc} (after all, one could use, for instance, norm-based ``metrics'') it is very convenient for the establishment of conceptual bridges with well-known information theoretic quantities.

Our second axiom aims at making explicit reference to measurements, another fundamental process through which an element of reality emerges. In a sense, this axiom is related to the first one in that a measurement can be viewed as a process whereby information about an observable is shared with an apparatus. On the other hand, a measurement is a very special instance involving, at the last stage, updating of information in the observer's mind, a physical system whose informational dynamics is often excluded from the theoretical description. For this reason, the quantum state collapse is generally used as an effective description for the measurement process. We can also envisage scenarios where measurements are  performed on an ensemble but the results are neither registered nor revealed to the observer. Let us consider such a nonselective measurement for a nondegenerate discrete-spectrum observable $A=\sum_i a_i A_i$, where $A_i=\ket{a_i}\bra{a_i}$ are projectors such that $A_iA_j=\delta_{ij}A_i$ and $\sum_i A_i=\mbb{1}_\mc{A}$. Since no  outcome is ever revealed, the post-measurement state is given by
\be\label{map_phi}
\Phi_A(\rho)\coloneqq\sum_i (A_i\otimes \mbb{1}_\mathcal{B})\,\rho\,(A_i\otimes \mbb{1}_\mc{B})=\sum_i p_i A_i\otimes\rho_{\mc{B}|i},
\ee
where $\rho_{\mc{B}|i}=\Tr_\mc{A}[(A_i\otimes \mbb{1}_\mc{B})\rho]/p_i$ and $p_i=\Tr[(A_i\otimes\mbb{1}_\mc{B})\rho]$. The nonselective-measurement map $\Phi_A$ has shown to be very useful in BA's approach. By considering a protocol involving the action of a secret agent, who intercepts the state, measures $A$ (thus creating an element of reality), and then submits the system to state tomography, BA concluded that the relation $\Phi_A(\rho)=\rho$ can be taken as an operational criterion of realism. The rationale behind this criterion is as follows. If the post-measurement state $\Phi_A(\rho)$ equals the pre-measurement state $\rho$, then the secret measurement has played no role whatsoever for the establishment of realism, meaning that $\rho$ already was a state for which $A$ is an element of reality. In this circumstance, the assessed state $\rho$ is termed an $A$-reality state\footnote{To further appreciate the reason why BA's criterion relies on nonselective instead of selective measurements, consider the $A$-reality state $\rho=(1-p)\ket{a_1}\bra{a_1}+p\ket{a_2}\bra{a_2}$. A selective measurement of $A$ would result in $\ket{a_i}\bra{a_i}$ ($i\in\{1,2\}$), implying a clear alteration in the description of the system and eventually inducing one to believe that the measurement has somehow changed the state of affairs. Conversely, using the nonselective-measurement map one obtains $\Phi_A(\rho)=\rho$, which shows that $\rho$ actually is immune to $A$ measurements and, therefore, can be viewed as a state in which $A$-reality is already installed.}. The idea of taking $\Phi_A(\rho)$ as an $A$-reality state can be further supported by the realization that it also describes a scenario wherein, after collapsing the description to $A_i\otimes\rho_{\mc{B}|i}$, the observer suddenly forgets the measurement outcome $a_i$. This misfortune, though, should not change the fact that an element of reality has just been established by the projective measurement. In addition, because the $A$ reality was already installed, repeating the procedure alters nothing, that is, $\Phi_A\left(\Phi_A(\rho)\right)=\Phi_A(\rho)$. We can also consider a {\it monitoring} of $A$ \cite{dieguez2018}, a generalized version of the unrevealed projective measurement \eqref{map_phi} that is able to interpolate weak and strong measurements through the strength parameter $\epsilon\in[0,1]$. Formally, the monitoring of $A$ is written as
\be\label{monitoring}
\mc{M}_A^\epsilon(\rho)\coloneqq (1-\epsilon)\,\rho+\epsilon\,\Phi_A(\rho).
\ee
Implementing a positive operator-valued measure (POVM) with effects $\{\sqrt{1-\epsilon}\mathbbm{1},\sqrt{\epsilon}A_i\}$, this map is expected to increase the reality of $A$ whenever $\epsilon > 0$. The second axiom then follows.

\begin{axiom}[Reality and measurements]\label{A_reality_meas}
The reality $\fR_A(\rho)$ is a nonnegative real number bounded from above by $\fR_\mc{A}^{\max}$. It is maximum \textit{iff} $\rho$ is an $A$-reality state and never decreases upon generalized measurements of $A$, that is, 
\be
0\leqslant\fR_A(\rho)\leqslant\fR_A\left(\mc{M}_A^\epsilon(\rho)\right)\leqslant \fR_A\left(\Phi_A(\rho)\right)\coloneqq\fR_\mc{A}^{\max},
\ee
where the second and third equalities hold iff $\Phi_A(\rho)=\rho$.
\end{axiom}
\noindent Given the informational nature of the reality quantifier $\fR_A$ and because the maximum amount of information a given Hilbert space can codify is bounded by its dimension, the upper bound $\fR_\mathcal{A}^{\max}$ is expected to depend on $d_\mc{A}=\dim(\mc{H_A})$. Note that while Axiom \ref{A_reality_inf} specifies the measure unity by which reality is quantified, Axiom \ref{A_reality_meas} establishes a numerical scale. The intrinsic relation between these axioms can be appreciated in terms of the dynamics imposed on the initial state $\rho\otimes\ket{e_0}\bra{e_0}$ by a given unitary operator $U_t$ acting on $\mc{H_S\otimes H_E}$. By use of the Stinespring dilation theorem (see also Ref.~\cite{dieguez2018}), we have
\be\label{stinespring}
\rho_t=\Tr_\mathcal{E}\left[U_t\left(\rho\otimes\ket{e_0}\bra{e_0} \right)U_t^\dagger \right]=\mc{M}_A^\epsilon(\rho),
\ee
where $\rho\in\mf{B}(\mc{H_S})$, $U_{t=0}=\mbb{1}_\mc{SE}$, and $\epsilon$ is a parameter that depends on $t$ and other characteristics of $U_t$.
The action of $\mc{M}_A^\epsilon$ generally changes the purity degree of $\rho$, so that some correlations with the environment and corresponding alterations in $I_\mc{E|S}$ are expected to occur, in agreement with the prescription~\eqref{DR=DI}.

To state our third axiom, we appeal to the intuition that the reality status of a physical quantity should not decrease upon discard or addition of uncorrelated degrees of freedom. On the other hand, we cannot exclude the possibility of increasing the realism of a quantity when the discarded system is correlated because in this case the system state undergoes an effective decoherence process, and hence, is shifted toward classical reality.
\begin{axiom}[Role of other parts]\label{A_role_parts}
\textrm{\bf (a)} Discarding a part of the system does not diminishes reality, that is,
\begin{subequations}
\be
\fR_A\left(\Tr_\mc{X}(\rho)\right)\geqslant \fR_A(\rho),
\ee
for $\mc{H_X}\subseteq\mathcal{H_B}$, where the equality applies when the discarded part is uncorrelated. Also, \textrm{\bf (b)} adding a fully uncorrelated system $\mc{Z}$ can by no means change the elements of reality of the system $\mc{S}$, that is
\be
\fR_A(\rho\otimes\Omega)=\fR_A(\rho),
\ee
\end{subequations}
where $\Omega\in\mf{B}(\mc{H}_\mc{Z})$.
\end{axiom}
\noindent From a mathematical viewpoint, we can recognize by Axioms \ref{A_reality_meas} and \ref{A_role_parts} a set of maps, henceforth called {\it realistic operations}, that do not diminish the reality of an observable. Formally, a realistic operation is a map $\rho \mapsto \Gamma(\rho)$ such that $\fR_A\left(\Gamma(\rho)\right)\geqslant \fR_A(\rho)$. For the above axioms, we have identified a particular set of realistic operations, that is, $\Gamma\in\{\mc{M}_A^\epsilon,\Tr_\mc{X},\otimes\,\Omega\}$.

With our fourth axiom we make a clear departure from classical reality. The point consists of implementing the intuition according to which, for a generic preparation $\rho$, noncommuting observables, such as three orthogonal spin components, cannot be simultaneous elements of reality. In other words, quantum realism is expected to be upper bounded.
\begin{axiom}[Uncertainty relation]\label{A_uncertainty}
Two observables $X$ and $Y$ acting on $\mc{H_A}$ cannot be simultaneous elements of reality in general, that is,
\be
\fR_X(\rho)+\fR_{Y}(\rho)\leqslant 2\fR^{\max}_\mc{A}.
\label{uncertainty}
\ee

\end{axiom}
The equality is expected to hold only in ``classical-like'' circumstances, such as  $\rho=(\mbb{1}/d_\mc{A})\otimes\rho_\mc{B}$ or $[X,Y]=0$. The  statement \eqref{uncertainty} links quantum realism to Bohr's complementarity principle. Interestingly, a recent experiment conducted in a nuclear magnetic resonance platform \cite{dieguez2021} has been reported corroborating the validity of the uncertainty relation \eqref{uncertainty} within the information theory induced by the von Neumann entropy.

Let us consider now a collection of quantum states $\rho_i\in\mc{H_S}$ with associated probabilities $p_i$ and realities $\fR_A(\rho_i)$. We do not expect the simple combination of these individual members to generate an ensemble with a lower reality status. In fact, mixing typically is an action toward classicality, so that reality is expected to be a concave functional. The fifth axiom is then stated as follows.
\begin{axiom}[Mixing]\label{A_mix}
The reality of a mixture $\{p_i,\rho_i\}$ of density operators $\rho_i$ with respective weights $p_i$ can never decrease the installed mean reality, that is,
\be
\fR_A\left(\sum_i p_i\rho_i \right)\geqslant \sum_i p_i\fR_A(\rho_i).
\ee
\end{axiom}
%

So far, we have presented the properties that we consider sufficient to define a meaningful {\it reality monotone}, in the sense that, upon the processes described above, reality never decreases\footnote{With respect to Axiom \ref{A_reality_inf}, we are of course envisaging dynamics whereby correlations typically build up so that $\Delta I_\mc{E|S}(t_2,t_1)\geqslant 0$. This is particularly true when the environment $\mc{E}$ is a genuine reservoir, like a thermal bath.}. That is, the typical move is toward classical reality, not the opposite. Although this set of axioms is rather constraining, we shall see in Sec. \ref{secMonotones} that it can be satisfied by a number of quantifiers supported not only by the standard von Neumann information theory but also by the R\'enyi and the Tsallis ones. This justifies the following definition.

\begin{definition}
A functional $\rho\mapsto\fR_A(\rho)$ satisfying Axioms \ref{A_reality_inf}-\ref{A_mix} is called a \textit{reality monotone}.
\end{definition}
\noindent 

In what follows we introduce two supplementary properties that can arguably be viewed as natural requirements for a reality measure.

\begin{axiom}[Additivity]\label{A_add}
The reality is an additive quantity over $n$ independent systems each one prepared in a state $\rho_i$, that is,
\be\label{eq_A_add}
\fR_A\left(\bigotimes_{i=1}^n\rho_i\right)=\sum_{i=1}^n\fR_A(\rho_i),
\ee
where $A$, on the left-hand side, acts on each one of the $n$ systems.
\end{axiom}
\noindent In particular, this means that given $n$ independent (eventually far apart) systems prepared in the same state $\rho$, the total amount of reality of an observable $A$ that acts on each $\rho$ is nothing but the direct sum $n\fR_A(\rho)$.

\begin{axiom}[Flagging]\label{A_flag}
The mean reality of an ensemble $\{p_i,\rho_i\}$ does not change under flagging, that is, 
\be\label{flagging}
\fR_A\left(\sum_i p_i\rho_i\otimes \ket{x_i}\bra{x_i} \right)= \sum_i p_i\fR_A(\rho_i).
\ee
\end{axiom}
\noindent The flagging property \cite{liu2019} has recently been discussed within the context of quantum resource theories. Suppose one identifies with a flag $\ket{x_i}\in\mc{H_X}$ each one of the states $\rho_i\in\mc{H_S}$ of our collection. The above axiom reflects the fact that merely labeling each element of the ensemble with a flag  basis $\{\ket{x_i}\}$ should not increase the mean reality. In other words, the insertion of classical correlations with respect to the flag is innocuous on average. 

With the above axioms we have set the grounds to define what we propose to be a significant reality quantifier.
\begin{definition}\label{def_Rmeasure}
A functional $\rho\mapsto\fR_A(\rho)$ satisfying Axioms \ref{A_reality_inf}-\ref{A_flag} is called a \textit{reality measure}.
\end{definition}

The next section is reserved for a brief review of information theoretic quantities that will be the basis for the construction of faithful reality quantifiers.

\section{Elements of quantum information theory}
\label{secElements}

Quantum divergences (or relative entropies) are measures of the distinctiveness of positive operators. These measures are known for their usefulness and versatility in defining several quantum information concepts, in particular, the one that will be shown to be of key relevance in this work, namely, the quantum conditional information. We now review three divergence measures, namely, the von Neumann relative entropy, R\'enyi divergences, and Tsallis relative entropies.

\subsection{von Neumann relative entropy}

The von Neumann relative entropy, also known as the Umegaki relative entropy \cite{umegaki1962}, is one of the most used divergences in quantum information theory. It is defined as
\be\label{vonNeuman}
D(\rho||\sigma)\coloneqq\frac{\Tr\left[\rho(\ln\rho-\ln\sigma)\right]}{\Tr \rho},
\ee
where $\rho>0$, $\sigma\geqslant 0$ and $\ker\sigma\subseteq\ker\rho$. The factor $\Tr\rho$ ensures that $D(\lambda\rho||\lambda\sigma)=D(\rho||\sigma)$ for all $\lambda>0\in\mathbb{R}$. $D(\rho||\sigma)$ is a continuous functional satisfying (whenever $\rho\geqslant\sigma $) the positive definiteness property:
\be\label{D>0}
D(\rho||\sigma)\geqslant 0,\ \text{with equality holding iff}\ \rho=\sigma.
\ee
The von Neumann relative entropy also satisfies the following properties: (i) unitary invariance,
\be\label{uni_invariance}
D\left(U\rho U^\dagger||U\sigma U^\dagger\right)=D(\rho||\sigma),
\ee
for any unitary $U$; (ii) additivity,
\be\label{additivity}
D\left(\bigotimes_i\rho_i\bigdbar\bigotimes_i\sigma_i\right)=\sum_i D(\rho_i||\sigma_i);
\ee
(iii) joint convexity,
\be\label{convexity}
D\left(\sum_i p_i\rho_i\bigdbar\sum_i p_i\sigma_i\right)\leqslant \sum_i p_i D(\rho_i||\sigma_i);
\ee
and (iv) data processing inequality (DPI),
\be\label{DPI}
D\left(\Lambda(\rho)||\Lambda(\sigma)\right)\leqslant D(\rho||\sigma),
\ee
also known as contractivity or monotonicity under quantum channels $\Lambda$ (completely positive trace-preserving maps).

The largest divergence implied by Eq. \eqref{vonNeuman} emerges when one considers a generic pure state, $\psi=\ket{\psi}\bra{\psi}$, and the maximally mixed one, $\mbb{1}/d$, with $d=\dim\mathcal{H}$. We have $D\left(\psi||\mbb{1}/d\right)=S(\mbb{1}/d)=\ln{d}$ (sometimes referred to as the normalization condition $D(\mbb{1}||\mbb{1}/d)=S(\mbb{1}/d)$ \cite{lennert2013}), where 
\be
S(\rho)\coloneqq -\frac{\Tr(\rho\ln\rho)}{\Tr \rho}
\ee 
is the von Neumann entropy of $\rho$. The quantum informational content $I(\rho)$ of a quantum state $\rho$ is a concept complementary to ignorance, that is, $I(\rho)+S(\rho)=I^{\max}=S^{\max}$ with $S^{\max}=S(\mbb{1}/d)=\ln{d}=I(\psi)=I^{\max}$ (meaning that the entropy of a maximally mixed state $\mbb{1}/d$ equals the informational content of a pure state $\psi$). In terms of the relative entropy, information can be defined as
\be\label{inf_content}
I(\rho)\coloneqq D(\rho||\mathbbm{1}/d)=\ln{d}-S(\rho),
\ee
Since pure states (resp. maximally mixed states) have maximum (resp. minimum) informational content, $I(\rho)$ is itself a direct measure of purity. One can make a further interpretation of $I(\rho)$ referring back to the map \eqref{map_phi}. Consider the pairs  $\{A,A'\}$ and $\{B,B'\}$ of noncommuting operators acting on $\mc{H_A}$ and $\mc{H_B}$, respectively, and forming maximally unbiased bases (MUB). One has $\Phi_{AA'}(\rho)\equiv\Phi_A\Phi_{A'}(\rho)=\Phi_{A'}\Phi_A(\rho)=\frac{\mbb{1}_\mc{A}}{d_\mc{A}}\otimes \rho_\mc{B}$, where $\rho_\mc{B}=\Tr_\mc{A}(\rho)$, and similarly for $\{B,B'\}$. For the whole context $\mbb{C}=\{A,A',B,B'\}$, we can write  $\Phi_\mbb{C}(\rho)=\frac{\mbb{1}_\mc{A}}{d_\mc{A}}\otimes\frac{\mbb{1}_\mc{B}}{d_\mc{B}}=\mbb{1}/d$, with $d=d_\mc{A}d_\mc{B}$. This is a state of full reality (or classical reality), since for any observable $X$ one has $\Phi_X(\mbb{1}/d)=\mbb{1}/d$, that is, a nonselective measurement of $X$ cannot change the established state of affairs. Therefore, we can rewrite Eq.~\eqref{inf_content} in the form $I(\rho)=D(\rho||\Phi_\mbb{C}(\rho))$, which allows us to interpret the informational content as the divergence of $\rho$ with respect to its classical counterpart $\Phi_\mbb{C}(\rho)$.

Equation \eqref{vonNeuman} can also be used to define the quantum conditional entropy of a quantum state $\rho$,
\be\label{H_vonNeuman}
H_\mc{A|B}(\rho)\coloneqq -D(\rho||\mbb{1}_\mc{A}\otimes\rho_\mc{B}).
\ee
It can be checked that this formula yields the usual relation $H_\mc{A|B}(\rho)=S(\rho)-S(\rho_\mc{B})$. The conditional entropy can alternatively be defined through an optimization process over the subspace $\mc{B}$, since $\inf_{\sigma_\mc{B}}D(\rho||\mbb{1}_\mc{A}\otimes\sigma_\mc{B})=D(\rho||\mbb{1}_\mc{A}\otimes\rho_\mc{B})$. By its turn, the conditional information of $\rho$ can also be defined through the {\it information-ignorance complementarity}, that is, $I_\mc{A|B}(\rho)+H_\mc{A|B}(\rho)=H_\mc{A|B}^{\max}=H_\mc{A|B}\left(\frac{\mbb{1}_\mc{A}}{d_\mc{A}}\otimes\rho_\mc{B}\right)=\ln{d_\mc{A}}$, a relation that will be taken as a fundamental premise in all information theories throughout this article. We then write
\be\label{Icond_vonNeuman}
I_\mc{A|B}(\rho)\coloneqq \ln{d_\mc{A}}-H_{\mc{A|B}}(\rho)=D\left(\rho\bigdbar \tfrac{\mbb{1}_\mc{A}}{d_\mc{A}}\otimes\rho_\mc{B} \right).
\ee
Because both entries in the above divergence are normalized density operators, one has $0\leqslant I_{\mathcal{A}|\mathcal{B}}(\rho)\leqslant \ln d$. Also, the conditional information can be decomposed as $I_\mc{A|B}(\rho)=I(\rho_\mc{A})+I_\mc{A:B}(\rho)$, where $I(\rho_\mc{A})=D\left(\rho_\mc{A}||\mbb{1}_\mc{A}/d_\mc{A}\right)$ is the informational content of part $\mc{A}$ and $I_\mc{A:B}(\rho)=D\left(\rho||\rho_\mc{A}\otimes\rho_\mc{B}\right)$ is the mutual information, a measure of total correlations between the parts. In this sense, $I_\mc{A|B}$ can be said to be composed of ``local'' and ``global'' information. Now, using the state $\upsilon_t=U_t\left(\rho\otimes\ket{e_0}\bra{e_0}\right)U_t^\dag$, we can compute the variation of $I_\mc{E|S}(\upsilon_t)=I\left(\Tr_\mc{S}(\upsilon_t)\right)+I_\mc{E:S}(\upsilon_t)$ in the interval $[0,t]$ and then return to Axiom \ref{A_reality_inf} to better specify the notion of ``information flow''. The condition for the $A$-reality increase, $\Delta I_\mc{E|S}(t,0)>0$, will be satisfied when $I_\mc{E:S}(\upsilon_t)>S\left(\Tr_\mc{S}(\upsilon_t)\right)$, which means that the share of information (correlations) between system and environment has to be sufficiently large for the emergence of reality. An alternative way of appreciating the role of the information flow for the emergence of realism is by writing $I(\upsilon_t)=I\big(\Tr_\mc{E}\upsilon_t\big)+I_\mc{E|S}(\upsilon_t)$ and then noticing that $I(\upsilon_t)$ is conserved in any unitary dynamics. It readily follows that $\Delta I_\mc{E|S}=I\big(\Tr_\mc{E}\upsilon_0\big)-I\big(\Tr_\mc{E}\upsilon_t\big)\equiv-\Delta I_\mc{S}$. By Axiom \ref{A_reality_inf} we then have $\Delta\mf{R}_A=-\Delta I_\mc{S}$, which shows that the $A$-reality increases whenever information ``flows out of the system''.

\subsection{R\'enyi divergences}

Constituting a generalization of the von Neumann relative entropy, the R\'enyi divergences \cite{petz1986} are defined as
\be\label{renyi}
D_\alpha(\rho||\sigma)\coloneqq\frac{1}{\alpha-1}\ln\frac{\Tr\left(\rho^\alpha\sigma^{1-\alpha}\right)}{\Tr\rho},
\ee
for $\alpha\in(0,1)\cup(1,+\infty)$ and the same conditions of quantity \eqref{vonNeuman}. Here, we also have $D_\alpha(\lambda\rho||\lambda \sigma)=D_\alpha(\rho||\sigma)$, for any positive real $\lambda$. Equation \eqref{renyi} is said a generalization of Eq. \eqref{vonNeuman} because $D_{\alpha\to 1}(\rho||\sigma)=D(\rho||\sigma)$. Another relative entropy comprised by the R\'enyi divergences is the min-relative entropy $D_{\min}(\rho||\sigma)\coloneqq \lim_{\alpha\to 0}D_{\alpha}(\rho||\sigma)=-\ln[\Tr(\rho^0\sigma)/\Tr\rho]$ where $\rho^0$ is the projection onto the support of $\rho$, as defined by Datta \cite{datta2009} (see Table \ref{tab:properties} for a summary of properties and Appendix \ref{appendix_cases} for more details about the min-relative entropy). Some properties that are satisfied by the von Neumann relative entropy encounter, however, some restrictions in the R\'enyi generalization: joint convexity is valid only for $\alpha\in(0,1)$ and DPI only for $\alpha\in(0,1)\cup(1,2]$ (see Ref. \cite{mosonyi2011} and references therein). The rest of the properties remain intact. A variant of definition \eqref{renyi} is the so-called sandwiched R\'enyi divergence, which was independently proposed by M\"uller-Lennert {\it et al.} \cite{lennert2013} and Wilde {\it et al.} \cite{wilde2014} as
\be\label{sandwiched}
\tD_\alpha(\rho||\sigma)\coloneqq\frac{1}{\alpha-1}\ln\left\{\frac{1}{\Tr\rho}\Tr\left[\left(\sigma^\frac{1-\alpha}{2\alpha}\rho\sigma^\frac{1-\alpha}{2\alpha} \right)^\alpha \right]\right\},
\ee
with the same conditions of quantity \eqref{renyi}. Besides reducing to the von Neumann relative entropy as $\alpha\to 1$, the divergence \eqref{sandwiched} reproduces other famous relative entropies, such as the collisional relative entropy ($\alpha=2$) \cite{renner2005} and the max-relative entropy $D_{\max}(\rho||\sigma)\coloneqq \lim_{\alpha\to+\infty}\tD_\alpha(\rho||\sigma)$ \cite{datta2009} (see Table \ref{tab:properties} for a summary of properties and Appendix \ref{appendix_cases} for more details about the collisional and the max-relative entropies). The sandwiched R\'enyi divergence  satisfies the same properties as its counterpart \eqref{renyi} but for different ranges of parameters: joint convexity is valid only for $\alpha\in[1/2,1)$ while DPI holds for $\alpha\in[1/2,1)\cup(1,+\infty)$ \cite{frank2013,beigi2013}. It was proved for $\alpha\in(0,1)$ \cite{datta2014} and $\alpha>1$ \cite{wilde2014} that the inequality $\tD_\alpha(\rho,\sigma)\leqslant D_\alpha(\rho,\sigma)$ is always true, where the equality holds iff $[\rho,\sigma]=0$. The necessity of this statement was noted in Ref. \cite{mosonyi2015}.

The issue concerning the commutativity of operators raised the discussion about the use of the divergence \eqref{sandwiched} instead of \eqref{renyi}. However, as pointed out by Gupta and Wilde \cite{gupta2015}, $D_\alpha(\rho||\sigma)$ ``is perfectly well defined'' when $\rho$ and $\sigma$ do not commute and, in fact, this divergence has proven to be useful for discrimination tasks in some contexts when $\alpha\in(0,1)$, including the limiting case $\alpha\to 0$ (see Ref. \cite{gupta2015} and references therein). The problem with definition \eqref{renyi} is that it does not satisfy DPI for $\alpha\in(2,+\infty)$, a large range that is in fact covered by the sandwiched version, including its limiting case ($\alpha\to +\infty$) known as max-relative entropy \cite{renner2005}. Since divergences are fundamental tools for one to distinguish a quantum state from another, it is expected that after the action of a quantum channel the states become less distinguishable and, therefore, DPI is an essential property for quantum information. Nonetheless, to obtain reality quantifiers it will be sufficient for us to focus on the original version of the R\'enyi divergence, since all the results will directly have a counterpart for the sandwiched version.

Now, using Eq. \eqref{renyi} one checks the validity of the normalization condition, $D_\alpha(\psi||\mbb{1}/d)=\ln{d}=S_\alpha(\mbb{1}/d)$, where
\be
S_\alpha(\rho)\coloneqq -\frac{1}{\alpha-1}\ln\frac{\Tr\rho^\alpha}{\Tr\rho},
\ee
is the quantum R\'enyi entropy of $\rho$. The R\'enyi informational content of $\rho$ can be defined as
\be
I_\alpha(\rho)\coloneqq D_\alpha(\rho||\mathbbm{1}/d)=\ln{d}-S_\alpha(\rho),
\ee
which reproduces Eq.~\eqref{inf_content} as $\alpha\to 1$. Since $\rho$ commutes with $\mathbbm{1}/d$, the original and the sandwiched R\'enyi divergences result in the same informational content. Here as well, we can interpret the informational content as the amount by which $\rho$ diverges from a full reality state, that is, $I_\alpha(\rho)=D_\alpha\left(\rho||\Phi_\mbb{C}(\rho)\right)$.

It is usual to define the R\'enyi conditional entropy in at least two different ways:
\begin{subequations}
\begin{align}
H_{\mc{A|B}}^{\alpha\down}(\rho)&\coloneqq -D_\alpha(\rho||\mbb{1}_\mc{A}\otimes\rho_\mc{B}),\label{Ha}\\
H_{\mc{A|B}}^{\alpha\up}(\rho)&\coloneqq -\inf_{\sigma_\mc{B}}D_\alpha(\rho||\mbb{1}_\mc{A}\otimes\sigma_\mc{B}),\label{Hb}
\end{align}
\end{subequations}
with $\sigma_\mc{B}\in\mf{B}(\mc{H_B})$.
The arrows are used to express the relation $H_{\mathcal{A|B}}^{\alpha\up}\geqslant H_{\mathcal{A|B}}^{\alpha\down}$.
It is noteworthy that, unlike its von Neumann counterpart \eqref{H_vonNeuman}, the R\'enyi conditional entropy cannot be expanded as $S_\alpha(\rho)-S_\alpha(\rho_\mc{B})$. Moreover, as emphasized by Tomamichel {\it et al}. \cite{tomamichel2014}, proposals along these lines lead to conceptual problems, such as the invalidation of DPI. From the complementarity relation $I_\mc{A|B}^{\alpha\up,\down}(\rho)+H_\mc{A|B}^{\alpha\up,\down}(\rho)=[H_{A|B}^\alpha]_{\max}=H_\mc{A|B}^{\alpha\up,\down}\left(\frac{\mbb{1}_\mc{A}}{d_\mc{A}}\otimes \rho_\mc{B}\right)=\ln{d_\mc{A}}$, we propose the R\'enyi conditional information measures
\begin{subequations}
\begin{align}
I_\mc{A|B}^{\alpha\up}(\rho)\coloneqq \ln d_\mc{A}- H_\mc{A|B}^{\alpha\down}(\rho),\label{Iup}\\
I_\mc{A|B}^{\alpha\down}(\rho)\coloneqq \ln d_\mc{A}- H_\mc{A|B}^{\alpha\up}(\rho),\label{Idown}
\end{align}
\end{subequations}
with arrows justified by the relation $I_\mc{A|B}^{\alpha\down}(\rho)\leqslant I_\mc{A|B}^{\alpha\up}(\rho)$. For any quantum channel $\Lambda_{\mc{B}\to\mc{B'}}$, both  measures satisfy DPI, that is, $  I_{\mc{A|B}}^{\alpha\up,\down}(\Lambda(\rho))\leqslant I_{\mc{A|B'}}^{\alpha\up,\down}(\rho)$ 
for $\alpha\in(0,1)\cup(1,2]$ (including the limiting case $\alpha\to 0$) and $\alpha\in [1/2,1)\cup(1,+\infty)$, respectively \cite{tomamichel2014}. Both conditional information measures are  convex under mixing for $\alpha\in(0,1)$ and $\alpha\in[1/2,1)$, respectively \cite{tomamichel2015}. Finally, by use of Eq. \eqref{renyi} we have 
\begin{subequations}
\begin{align}
I_\mc{A|B}^{\alpha\up}(\rho)&=D_\alpha\left(\rho\bigdbar\tfrac{\mbb{1}_\mc{A}}{d_\mc{A}}\otimes\rho_\mc{B}\right),\label{Icond_Renyi_up}\\
I_\mc{A|B}^{\alpha\down}(\rho)&=\inf_{\sigma_\mc{B}}D_\alpha\left(\rho\bigdbar\tfrac{\mbb{1}_\mc{A}}{d_\mc{A}}\otimes\sigma_\mc{B}\right).\label{Icond_Renyi_down}
\end{align}\label{Icond_Renyi}
\end{subequations}

\begin{table*}[htb]
\renewcommand{\arraystretch}{1.5}
\centering
\begin{tabular}{p{0.15\textwidth}>{\centering}p{0.12\textwidth}>{\centering}p{0.12\textwidth}>{\centering}p{0.12\textwidth}>{\centering}p{0.15\textwidth}>{\centering}p{0.12\textwidth}cp{0.12\textwidth}}
\hline
    & $D$ & $D_\alpha$ & $D_{\min}$ & $\tD_\alpha$ & $D_{\max}$ & $D_q$ \\
\hline
    Continuity & \cmark & \cmark & \xmark & \cmark & \xmark& \cmark \\
    Positive definiteness & \cmark & \cmark & \xmark& \cmark & \cmark & \cmark \\
    Unitary invariance & \cmark & \cmark & \cmark & \cmark & \cmark & \cmark\\
    Additivity & \cmark & \cmark & \cmark & \cmark & \cmark & \xmark\\
    Joint convexity & \cmark & $\alpha\in(0,1)$ & \cmark & $\alpha\in[1/2,1)$ & \xmark& $q\in(0,1)\cup(1,2]$ \\
    DPI & \cmark & $\alpha\in(0,1)\cup(1,2]$ & \cmark & $\alpha\in[1/2,1)\cup(1,+\infty)$ & \cmark & $q\in(0,1)\cup(1,2]$\\
\hline
\end{tabular}
\caption{Summary of properties satisfied by the von Neumann relative entropy $D$, the R\'enyi divergence $D_\alpha$, the min-relative entropy $D_{\min}\coloneqq D_{\alpha\to 0}$, the sandwiched R\'enyi divergence $\tD_\alpha$, the max-relative entropy $D_{\max}\coloneqq \tD_{\alpha\to+\infty}$, and the Tsallis relative entropy $D_q$, for any pair $\{\rho,\sigma\}$ of density operators.}
\label{tab:properties}
\end{table*}

\subsection{Tsallis relative entropies}

To close this section, let us revisit the Tsallis relative entropies, originally proposed by Abe \cite{abe2003}. Here we adopt the form
\be\label{tsallis}
D_q(\rho||\sigma)\coloneqq\frac{\Tr\left[\rho^q\left(\ln_q {\rho}-\ln_q{\sigma}\right) \right]}{\Tr\rho}=\frac{\Tr(\rho-\rho^q\sigma^{1-q})}{(1-q)\,\Tr\rho},
\ee
where $q\in(0,1)$ and $\ln_q(x)\coloneqq (x^{1-q}-1)/(1-q)$. As pointed out by Rastegin \cite{rastegin2016}, the definition \eqref{tsallis} can be extended to $q>1$ if $\ker\sigma\subseteq\ker\rho$. The normalization guarantees that $D_q(\lambda\rho||\lambda\sigma)=D_q(\rho||\sigma)$ for any $\lambda>0\in\mbb{R}$. When $q\to 1$, we regain the von Neumann relative entropy. The Tsallis relative entropies and the R\'enyi divergences share several properties. $D_q(\rho||\sigma)$ is a continuous and positive definite functional in $\rho$ and $\sigma$ for $q\in(0,1)\cup(1,+\infty)$. In addition, the Tsallis relative entropies satisfy unitary invariance for $q\in(0,1)\cup(1,+\infty)$, and both joint convexity and DPI for $\alpha\in(0,1)\cup(1,2]$ \cite{hiai2011,rastegin2016}. Most importantly, they are pseudo-additive, that is
\begin{eqnarray}
&D_q(\rho_\mc{A}\otimes\rho_\mc{B}||\sigma_\mc{A}\otimes\sigma_\mc{B})=D_q(\rho_\mc{A}||\sigma_\mc{A})+D_q(\rho_\mc{B}||\sigma_\mc{B})&\nonumber\\
&+(q-1)D_q(\rho_\mc{A}||\sigma_\mc{A})D_q(\rho_\mc{B}||\sigma_\mc{B}).&\label{pseudo_add}
\end{eqnarray}
From Eq.~\eqref{tsallis} we find $D_q(\psi||\mbb{1}/d)=d^{q-1}S_q(\mbb{1}/d)$, where
\be
S_q(\rho)\coloneqq -\Tr\left(\rho^q\ln_q\rho\right)=-\frac{\Tr(\rho-\rho^q)}{(1-q)\,\Tr\rho}
\ee
is the Tsallis entropy of $\rho$ \cite{tsallis1988,raggio1995} and $S_q(\mbb{1}/d)=\ln_q{d}$. Note that, differently from the structure found for the previous information theories, here the normalization relation is such that $D_q(\psi||\mbb{1}/d)\neq S_q(\mbb{1}/d)$. This suggests that it may be convenient to ``correct'' either $D_q$ or $S_q$ by means of a scaling factor like $d^{1-q}$ or $d^{q-1}$. To preserve the fundamental status of the information-ignorance complementarity, we then define the Tsallis informational content as
\be
I_q(\rho)\coloneqq d^{1-q}D_q(\rho||\mbb{1}/d)=\ln_q{d}-S_q(\rho).
\ee
Similarly to what can be found in Ref. \cite{rajagopal2014}, let us define the Tsallis conditional entropy as
\be
H_\mc{A|B}^q(\rho)\coloneqq -D_q(\rho||\mbb{1}_\mc{A}\otimes\rho_\mc{B}).
\ee
One may wonder whether  $-\inf_{\sigma_\mc{B}}D_q(\rho||\mbb{1}_\mc{A}\otimes\sigma_\mc{B})$ would be an admissible formulation as well. Although we believe there is no reason why this proposal should be ruled out {\it a priori}, we are not aware of any study supporting it. Following previous rationales, we now look for a conditional information satisfying the information-ignorance relation $I_\mc{A|B}^q(\rho)+H_\mc{A|B}^q(\rho)=[H_\mc{A|B}^q]_{\max}=H_\mc{A|B}^q\left(\frac{\mbb{1}_\mc{A}}{d_\mc{A}}\otimes\rho_\mc{B}\right)=\ln_q{d_\mc{A}}$. We then find
\be
I_\mc{A|B}^q(\rho)\coloneqq\ln_q{d_\mc{A}}-H_\mc{A|B}^q(\rho),
\ee
Using the above formulas, one shows that
\be\label{Icond_Tsallis}
I_\mc{A|B}^q(\rho)=d_\mc{A}^{1-q}\,D_q\left(\rho\bigdbar\tfrac{\mbb{1}_\mc{A}}{d_\mc{A}}\otimes\rho_\mc{B}\right),
\ee
which correctly reduces to the result \eqref{Icond_vonNeuman} as $q\to 1$. Up to the scaling factor $d_\mc{A}^{1-q}$, proved necessary in the present scenario, we note by Eqs. \eqref{Icond_vonNeuman}, \eqref{Icond_Renyi_up}, and \eqref{Icond_Tsallis} that it is possible to maintain a unified picture for the definition of the conditional informational of $\rho$ in terms of its divergence with respect to its full reality counterpart,  $\Phi_{AA'}(\rho)=\frac{\mbb{1}_\mc{A}}{d_\mc{A}}\otimes\rho_\mc{B}$. We refer the reader to Table \ref{tab:properties} for a summary of properties that are satisfied by each divergence presented in this section.

\section{Reality monotones and measures}
\label{secMonotones}

\subsection{von Neumann reality measure}

Through the relation $\fR_A(\rho_t)-\fR_A(\rho)=I_\mc{E|S}(\upsilon_t)-I_\mc{E|S}(\upsilon_0)$, Axiom \ref{A_reality_inf} links the emergence of realism in the system $\mc{S}$ with the acquisition of information by the environment $\mc{E}$. Our strategy here consists of starting with the uncorrelated state $\upsilon_0=\rho\otimes\ket{e_0}\bra{e_0}$ and searching for a dynamics that yields maximum reality for $A$. That is, we want to find a reduced state $\rho_t=\Tr_\mc{E}(\upsilon_t)=\Phi_A(\rho)$ such that $\fR\left(\rho_t\right)=\fR_\mc{A}^{\max}$ as per Axiom \ref{A_reality_meas}, so that we can construct the $A$-reality measure $\fR_A(\rho)=\fR_\mc{A}^{\max}-\Delta I_\mc{E|S}(t,0)$. From the additivity [Eq. \eqref{additivity}] of the von Neumann conditional information [Eq. \eqref{Icond_vonNeuman}] we find
\begin{align}
I_\mc{E|S}(\upsilon_0)&=D\left(\rho\otimes\ket{e_0}\bra{e_0}\bigdbar\rho\otimes\tfrac{\mbb{1}_\mc{E}}{d_\mc{E}}\right)=\ln d_\mc{E}.\label{I0_vonNeumann}
\end{align}
Because there are no correlations in the initial state, the informational content of the environment is not conditioned to the system. Now we consider a dynamics induced by a unitary operator $U_t$ satisfying Eq. \eqref{stinespring} with $\epsilon=1$. Since $\Tr_\mc{E}(\upsilon_t)=\Phi_A(\rho)$, we have $ I_\mc{E|S}(\upsilon_t)=D\left(U_t\upsilon_0 U_t^\dag||\Phi_A(\rho)\otimes\mbb{1}_\mc{E}/d_\mc{E}\right)$. Using the fact that $\Phi_A(\rho)\otimes\mbb{1}_\mc{E}/d_\mc{E}$ does not evolve under the action of $U_t$ [see Theorem \ref{theorem_uni}, Appendix \ref{appendixB}], we can freely apply $U_t$ onto it and then use the unitary invariance of the von Neumann relative entropy to obtain $I_\mc{E|S}(\upsilon_t) =D\left(\upsilon_0 ||\Phi_A(\rho)\otimes\mbb{1}_\mc{E}/d_\mc{E}\right)$. Using additivity again, we get
\be
I_\mc{E|S}(\upsilon_t)
=\ln d_\mc{E}+D\left(\rho||\Phi_A(\rho)\right).
\label{It_vonNeumann}
\ee
From Eqs. \eqref{I0_vonNeumann} and \eqref{It_vonNeumann} we have $\Delta I_\mc{E|S}(t,0)=D\left(\rho||\Phi_A(\rho)\right)$ and hence $\fR_A(\rho)=\fR_\mc{A}^{\max}-D\left(\rho||\Phi_A(\rho)\right)$. We now use the fact that $D\left(\rho||\Phi_A(\rho)\right)\leqslant \ln{d_\mc{A}}$
(see Lemma \ref{lemma_bound}, Appendix \ref{appendixB}), to set $\fR_\mc{A}^{\max}=\ln{d_\mc{A}}$. This yields the reality quantifier
\be\label{R_vonNeumann}
\fR_A(\rho)=\ln d_\mc{A}-D\left(\rho||\Phi_A(\rho)\right),
\ee
which is such that $\fR_A(\rho)\geqslant 0$ and $\fR_A\left(\Phi_A(\rho)\right)=\ln{d_\mc{A}} $, as required by Axiom \ref{A_reality_meas}. A particularly interesting property of the reality quantifier \eqref{R_vonNeumann} is that it allows us to formally state a complementarity relation. To see this, we can employ Lemmas \ref{lemma_f} and \ref{lemma_bound} (see Appendix \ref{appendixB}) to demonstrate that $D(\rho||\Phi_A(\rho))=S(\Phi_A(\rho))-S(\rho)\eqqcolon \fI_A(\rho)$, where $\fI_A(\rho)$ is the \textit{irreality} (indefinite reality) of the observable $A$ given the state $\rho$, as originally proposed by BA \cite{bilobran2015}. We then have
\be\label{RandI}
\fR_A(\rho)+\fI_A(\rho)=\ln d_\mathcal{A}.
\ee
It becomes clear now the duality between irreality---a quantum resource {\it per se} \cite{costa2020}---and reality, which can thus be viewed as the amount of quantum resource that is destroyed when an observable is measured.

Now we show that the quantifier \eqref{R_vonNeumann} does satisfy Definition \ref{def_Rmeasure}, which characterizes it as a reality measure. Axiom \ref{A_reality_inf} was of course satisfied by construction. From DPI, we can check that any quantum channel $\Lambda$ that commutes with $\Phi_A$ for every $\rho$, that is, $\Lambda(\Phi_A(\rho))=\Phi_A(\Lambda(\rho))$, will never decrease the $A$ reality. This includes monitoring maps $\mc{M}_A^\epsilon$ of any intensity and the discarding of parts of the system that do not include $\mc{A}$. The Axioms \ref{A_reality_meas} and \ref{A_role_parts}(a) are therefore satisfied. Along with the fact that $\Phi_A(\rho\otimes\Omega)=\Phi_A(\rho)\otimes\Omega$, additivity [Eq.~\eqref{additivity}] guarantees that the quantifier \eqref{R_vonNeumann} satisfies the Axiom \ref{A_role_parts}(b). Therefore, we have $\fR_A\left(\Gamma(\rho)\right)\geqslant \fR_A(\rho)$, confirming that realistic operations $\Gamma$ cannot make realism decrease. That Axiom \ref{A_uncertainty} is respected follows from Lemma \ref{lemma_uncertainty} (Appendix \ref{appendixB}). 
At this point, it is opportune to remark how quantum correlations influence the realism uncertainty relation (Axiom \ref{A_uncertainty}). It has been shown in Ref. \cite{bilobran2015} that $\fI_A(\rho)=\fI_A(\rho_\mc{A})+\mc{D}_A(\rho)$, where $\mc{D}_A(\rho)=I_\mc{A:B}(\rho)-I_\mc{A:B}\left(\Phi_A(\rho)\right)$ is the (nonoptimized) quantum discord associated with the observable $A$ and $I_\mc{A:B}=D(\rho||\rho_\mc{A}\otimes\rho_\mc{B})$ is the mutual information. Since $\mc{D}_A(\rho)\geqslant \min_A\mc{D}_A(\rho)\equiv \mf{D}_\mc{A}(\rho)$, where $\mf{D}_\mc{A}$ stands for the one-sided quantum discord, one can conclude that $\mc{D}_X(\rho)+\mc{D}_Y(\rho)\geqslant 2\mf{D}_\mc{A}(\rho)$, with equality holding, for instance, for product states. Combining $\fI_A(\rho_\mc{A})=D\left(\rho_\mc{A}||\Phi_A(\rho_\mc{A})\right)\geqslant 0$ with Eq. \eqref{RandI}, we can verify that $\fR_A(\rho)\leqslant \ln{d_\mc{A}}-\mc{D}_\mc{A}(\rho)$ and
\be 
\fR_X(\rho)+\fR_Y(\rho)\leqslant  2\left[\ln{d_\mc{A}}-\mf{D}_\mc{A}(\rho)\right].
\ee 
This shows that quantum correlations, as measured by quantum discord (entanglement for pure states) forbid $X$ and $Y$ to be simultaneous elements of reality. Accordingly, using a nuclear magnetic resonance platform and associating $X$ and $Y$ with wave- and particle-like observables, researchers have recently reported on an experiment where an entangled quantum system behaves neither as a wave nor as particle \cite{dieguez2021}. The validity of Axiom \ref{A_mix} (mixing) comes immediately from joint convexity  [Eq.~\eqref{convexity}]. With respect to Axiom \ref{A_add} (additivity), we should first note that $\fR_A(\rho^{\otimes n})\coloneqq\ln d_\mathcal{A}^n-D(\rho^{\otimes n}||\Phi_A(\rho)^{\otimes n})$, that is, $A$ is presumed to act over each copy of $\rho$. It then follows from the identity \eqref{additivity} that $\fR_A(\rho^{\otimes n})=n\fR_A(\rho)$. Last but not least, to verify the validity of Axiom \ref{A_flag} (flagging), we start with $\fR_A(\rho_f)=\ln{d_\mc{A}}-S(\Phi_A(\rho_f))+S(\rho_f)$ (see the proof of Lemma \ref{lemma_bound}, Appendix \ref{appendixB}), with the flagged state $\rho_f=\sum_ip_i\rho_i\otimes\ket{x_i}\bra{x_i}$. The joint entropy theorem yields $S(\rho_f)=H(\{p_i\})+\sum_i p_i S(\rho_i)$, where $H(\{p_i\})=-\sum_ip_i\ln p_i$ is the Shannon entropy of the distribution $p_i$ \cite{nielsen2000}. Direct calculations gives $\fR_A(\rho_f)=\sum_ip_i\fR_A(\rho_i)$ with $\fR_A(\rho_i)=\ln d_\mathcal{A} -S(\Phi_A(\rho_i))+S(\rho_i)$, which proves the point. With all that, it becomes established that the quantifier  \eqref{R_vonNeumann} does indeed satisfy Definition \ref{def_Rmeasure} and can hereafter be called a reality measure. 

Referring back to Axiom \ref{A_reality_meas}, it is worth discussing how the reality measure \eqref{R_vonNeumann} changes upon monitoring maps [Eq.~\eqref{monitoring}]. First, because the reality measure respects Axiom \ref{A_mix} (mixing), which ultimately is a statement of concavity, one can readily show that $\fR_A(\mc{M}_A^\epsilon(\rho))\geqslant (1-\epsilon)\,\fR_A(\rho)+\epsilon\,\fR_A\left(\Phi_A(\rho)\right)$. Then, by use of Eq.~\eqref{RandI} we arrive at
\be
\fR_A\left(\mc{M}_A^\epsilon(\rho)\right)-\fR_A(\rho)\geqslant\epsilon\,\fI_A(\rho).
\ee
This shows that a monitoring of $A$ always increases the $A$ reality as long as there is a nonzero amount of $A$ irreality \cite{dieguez2018}. Second and more surprising, it turns out that the monitoring of an observable $Y\in\mf{B}(\mc{H_A})$ never diminishes the reality of another observable $X\in\mf{B}(\mc{H_A})$, that is,
\be
\fR_X\left(\mc{M}_Y^\epsilon(\rho)\right)-\fR_X(\rho)\geqslant 0,
\ee
$\forall\,\epsilon\in[0,1]$, whenever the $X$ and $Y$ eigenstates form MUB. This is one of the main results of Ref. \cite{dieguez2018}, but the reader can find a simpler alternative proof of it based on DPI and mixing in Appendix \ref{appendixB} (see Lemma \ref{lemma_delta}). Notice that the above inequality generalizes Axiom \ref{A_reality_meas}.

Finally, it is worth noticing that the so-called local irreality, $\mf{I}_A(\rho_\mc{A})=S\big(\Phi_A(\rho_\mc{A})\big)-S(\rho_\mc{A})$, which relates to irreality through the formula $\mf{I}_A(\rho)=\mf{I}_A(\rho_\mc{A})+\mc{D}_A(\rho)$, is nothing but the measure known as {\it relative entropy of coherence}~\cite{baumgratz2014}, which has been acknowledged as a quantum resource~\cite{streltsov2017}. This shows that quantum irrealism is induced by both types of ``quantumness'', namely, quantum coherence and quantum correlations. In particular, in the absence of correlations, one has $\mf{I}_A(\rho_\mc{A}\otimes\rho_\mc{B})=\mf{I}_A(\rho_\mc{A})$, showing that coherence is sufficient to preclude classical reality. Within the coherence theory of multipartite settings, the irreality $\mf{I}_A(\rho)$ turns out to be equivalent to the concept known as {\it quantum-incoherent relative entropy}~\cite{chitambar2016}. These connections between quantum irrealism and quantum coherence measures just reinforce that $\mf{I}_A(\rho)$ is a sensible quantifier of the former concept, for quantum superposition (coherence) is the fundamental mechanism responsible for the departure of the natural behavior from classical reality.

\subsection{R\'enyi reality monotones}

We now derive a reality quantifier based on the nonoptimized conditional information \eqref{Icond_Renyi_up}. Because this quantity and the von Neumann relative entropy share properties such as positive definiteness, unitary invariance, and additivity, we can rigidly follow the steps of the precedent section, which amounts to using Theorem \ref{theorem_uni} and Axiom \ref{A_reality_inf}, to directly propose the R\'enyi reality quantifier
\be\label{R_renyi_down}
\fR_A^{\alpha\down}(\rho)=\ln d_\mc{A}-D_\alpha\left(\rho||\Phi_A(\rho)\right),
\ee
for $\alpha\in(0,1)\cup(1,+\infty)$. Since $\lim_{\alpha\to 1}\fR_A^{\alpha\down}(\rho)=\fR_A(\rho)$ for any $\rho$ and $A$, we have here an evident generalization of \eqref{R_vonNeumann} within the R\'enyi quantum information theory. Inspired by the results of the previous section, we have chosen $\fR_\mc{A}^{\max}=\ln d_\mathcal{A}$ to make the quantity \eqref{R_renyi_down} always nonnegative (in particular, for $\alpha\to 1$). 

As we show now, the quantifier \eqref{R_renyi_down} is a reality monotone only in the restricted range $\alpha\in(0,1)$. Axioms \ref{A_reality_meas} and \ref{A_role_parts}(a) are satisfied whenever DPI is valid, in this case, for $\alpha\in(0,1)\cup(1,2]$. Axioms \ref{A_role_parts}(b) and \ref{A_uncertainty} are validated by additivity and positive definiteness, respectively. Axiom \ref{A_mix}, however, holds only when $D_\alpha$ is jointly convex, that is, for $\alpha\in(0,1)$. This significantly restricts the domain wherein $\fR_A^{\alpha}$ can be termed a reality monotone. Although additivity guarantees the Axiom \ref{A_add} to be respected by the monotone \eqref{R_renyi_down}, there is no answer yet as to whether or not $D_\alpha$ satisfies flagging. Only in the affirmative case could we regard $\fR_A^{\alpha}$ as a reality measure for $\alpha\in(0,1)$.

Since the R\'enyi divergence is a monotonically increasing real function of $\alpha$, for all $\alpha>0$ and fixed density operators \cite{mosonyi2011}, the reality measure \eqref{R_renyi_down} is a monotonically decreasing real function of its parameter, meaning that
\be\label{monotonicity}
\fR_A^{\alpha\down}(\rho)\geqslant\fR_A^{\beta\down}(\rho)
\ee
for real nonnegative numbers $\alpha\leqslant\beta$. This entails that if $\fR_A^{\beta\down}(\rho)=\ln d_\mathcal{A}$ for some $\beta\geqslant 0$, meaning that $\rho=\Phi_A(\rho)$, then $\fR_A^{\alpha\down}(\rho)=\ln d_\mathcal{A}$ for every $\alpha\leqslant \beta$. If, in addition, $\beta\to 1$, then all R\'enyi reality monotones will numerically reach the maximum $\ln{d_\mc{A}}$. Therefore, although R\'enyi reality monotones with different parameters $\alpha$ disagree in value when applied to nonreal observables (those for which $\fR_A^{\alpha\down}(\rho)<\ln d_\mathcal{A}$), they do always agree about states of reality (see Example \ref{ex_mu} and respective Fig. \ref{fig:RRmu} below).

One of the consequences of the positive definiteness property---which does not hold when we use the min-relative entropy---is that $\fR_A^{\alpha\down}(\rho)=0$ if and only if $\rho=\sum_ip_iA_i\otimes\rho_{\mc{B}|i}=\Phi_A(\rho)$, which is a classical-quantum state with zero one-sided quantum discord. This means that the lack of quantum correlations is a condition necessary for the occurrence of at least one element of reality. On the other hand, classical reality manifests itself for the preparation $\rho=(\mbb{1}_\mc{A}/d_\mc{A})\otimes\rho_\mc{B}$, since in this case we have $\fR_A^{\alpha\down}(\rho)=\ln{d_\mc{A}}$ for any $A$.

Next, we present some case studies.

\begin{figure}[htb]
\centering
\includegraphics[width=0.9\linewidth]{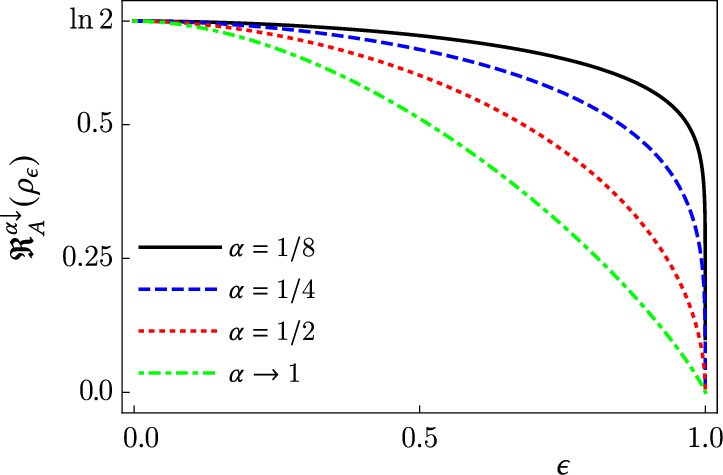}
\caption{Reality measure (dash-dotted green line) and R\'enyi monotones $\fR_A^{\alpha\down}(\rho_\epsilon)$ for any spin observable $A$ of the first qubit of a Werner state [Eq. \eqref{werner}] as a function of the purity parameter $\epsilon$ (as introduced in Example \ref{ex_werner}) for: $\alpha= 1/8$ (solid black line), $\alpha=1/4$ (dashed blue line), $\alpha= 1/2$ (dotted red line), and $\alpha\to 1$ (dash-dotted green line).}
\label{fig:RRwerner}
\end{figure}
\begin{figure*}[htb]
    \centering
    \includegraphics[width=0.9\linewidth]{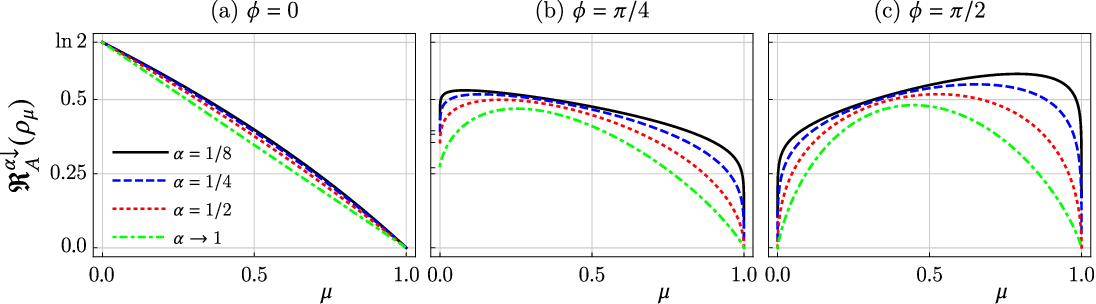}
    \caption{Reality monotones $\fR_A^{\alpha\down}(\rho_\mu)$ of the spin observable $A=\hat{u}\cdot\vec{\sigma}$, where $\hat{u}=(\cos\theta\sin\phi,\sin\theta\sin\phi,\cos\theta)$, regarding the first qubit of the $\rho_\mu$ state \eqref{mu_state} as a function of $\mu$ as introduced in Example \ref{ex_mu} for any $\theta$ and (a) $\phi=0$, (b) $\phi=\pi/4$, and (c) $\phi=\pi/2$ and for (from top to bottom): $\alpha= 1/8$ (solid black line), $\alpha=1/4$ (dashed blue line), $\alpha= 1/2$ (dotted red line), and $\alpha\to 1$ (dash-dotted green line).}
    \label{fig:RRmu}
\end{figure*}
\begin{example}\label{ex_werner}
Let $\rho_\epsilon\in\mathfrak{B}(\mathcal{H_A\otimes H_B})$ be the Werner state 
\be\label{werner}
\rho_\epsilon=(1-\epsilon)\tfrac{\mbb{1}}{4}+\epsilon\,\psi_s,
\ee
where $\epsilon\in[0,1]$, $\psi_s=\ket{\psi_s}\bra{\psi_s}$, and $\ket{\psi_s}=(\ket{01}-\ket{10})/\sqrt{2}$ is the singlet state. To assess the reality degree of the spin observable $A=\hat{u}\cdot\vec{\sigma}$ acting on $\mathcal{H_A}$, with $\vec{\sigma}=(\sigma_x,\sigma_y,\sigma_z)$ being the Pauli vector, we take the projectors $A_\pm=(\mbb{1}_\mc{A}\pm\hat{u}\cdot\vec{\sigma})/2$ with $\hat{u}=(\cos\theta\sin\phi,\sin\theta\sin\phi,\cos\phi)$ and then compute the $A$-reality state $\Phi_A(\rho_\epsilon)=A_+\rho_\epsilon A_++A_-\rho_\epsilon A_-$. Because $\rho_\epsilon$ is rotationally invariant, it commutes with $\Phi_A(\rho_\epsilon)$ and we have $\tD_\alpha(\rho_\epsilon||\Phi_A(\rho_\epsilon))=D_\alpha(\rho_\epsilon||\Phi_A(\rho_\epsilon))$ for $\alpha\in(0,1)\cup(1,+\infty)$. Therefore, both the original and the sandwiched R\'enyi divergences can be used within the range $\alpha\in(0,1)$ to provide a reality monotone for the Werner state. By direct calculation of $\fR_A^{\alpha\down}(\rho_\epsilon)=\ln 2-D_\alpha(\rho_\epsilon||\Phi_A(\rho_\epsilon))$ we find
\be\label{Rdown_werner}
\fR_A^{\alpha\down}(\rho_\epsilon)=\ln 2-
    \ln\left[\frac{(1-\epsilon)^\alpha+(1+3\epsilon)^\alpha}{4(1+\epsilon)^{\alpha-1}}+\frac{1-\epsilon}{2} \right]^\frac{1}{\alpha-1}.
\ee
Note that the special cases
\begin{subequations}
\be
\lim_{\alpha\to 0}\fR_A^{\alpha\down}(\rho_\epsilon)=\left\{ \begin{array}{ll}
    \ln 2 & \text{if $\epsilon\in[0,1)$ }, \\
    0 & \text{if $\epsilon=1$},
\end{array}\right.
\ee
\be\label{R_infinity}
\lim_{\alpha\to +\infty}\fR_A^{\alpha\down}(\rho_\epsilon)=\ln 2 -\ln\left(\frac{1+3\epsilon}{1+\epsilon} \right),
\ee
\end{subequations}
do not constitute reality monotones, since they do not satisfy Axioms \ref{A_reality_meas} and \ref{A_mix}, respectively. See Appendix \ref{appendix_cases} for the technical details on how to calculate the R\'enyi divergences when $\alpha\to 0$ and $+\infty$. As we can see in Fig. \ref{fig:RRwerner}, the monotonicity property \eqref{monotonicity} is verified. Also, since the R\'enyi monotone is concave (due to the mixing axiom) and $\fR_A^{\alpha\down}(\psi_s)=0$, then  $\fR_A^{\alpha\down}(\rho_\epsilon)\geqslant (1-\epsilon)\ln 2$. This result manifests itself in Fig. \ref{fig:RRwerner} through the concavity of the curves, a feature that is not respected by the convex function \eqref{R_infinity}.
\end{example}
\begin{example}\label{ex_mu}
Let us consider now the one-parameter two-qubit state $\rho_\mu\in\mathfrak{B}(\mathcal{H_A\otimes H_B})$ defined as
\be\label{mu_state}
\rho_\mu=\tfrac{\mbb{1}}{4}+\tfrac{\mu}{4}\left(\sigma_x\otimes\sigma_x-\sigma_y\otimes\sigma_y \right) + \tfrac{2\mu-1}{4}\sigma_z\otimes\sigma_z,
\ee
where $\mu\in[0,1]$. This state is such that $\rho_{\mu=1}=\ket{\varphi}\bra{\varphi}$, with $\ket{\varphi}\equiv(\ket{00}+\ket{11})/\sqrt{2}$, and  $\rho_{\mu=0}=(\ket{01}\bra{01}+\ket{10}\bra{10})/2$. Unlike the previous example, for the observable $A=\hat{u}\cdot\vec{\sigma}$ it follows that $\rho_\mu$ does not always commute with $\Phi_A(\rho_\mu)$. Indeed, our calculations show that $\fR_A^{\alpha\down}(\rho_\mu)$ (whose lengthy and nonenlightening expression will be omitted) depends on the polar angle $\phi$, as shown in Fig. \ref{fig:RRmu}. Again, the monotonicity relation \eqref{monotonicity} makes itself clear. Also, numerical simulations show that in this case a reality monotone based on $\tD_\alpha$ for $\alpha\in[1/2,1)$ behaves very similarly to what is presented in Fig. \ref{fig:RRmu}.
\end{example}

Now we turn our attention to the optimized version \eqref{Icond_Renyi_down} of the R\'enyi conditional information. Here the derivation of a reality monotone becomes subtler because the optimization process does not keep a unique and straightforward connection with the self-contained dynamical scenario prescribed by Axiom \ref{A_reality_inf}. Still, some potential candidates can be proposed. Let us consider again the initial state $\upsilon_0$. Plugged into Eq. \eqref{Icond_Renyi_down}, it yields $I_\mc{E|S}^{\alpha\down}(\upsilon_0)= \inf_{\sigma_\mc{S}}D_\alpha\left(\rho\otimes\ket{e_0}\bra{e_0}||\sigma_\mc{S}\otimes\mbb{1}_\mc{E}/d_\mc{E}\right)=\ln{d_\mc{E}}$, where we have used additivity and $\inf_{\sigma_\mathcal{S}}D_\alpha(\rho||\sigma_\mathcal{S})=0$. By applying $U_t$, we find $I_\mc{E|S}^{\alpha\down}(\upsilon_t)= \inf_{\sigma_\mc{S}}D_\alpha\left(\upsilon_t||\sigma_\mc{S}\otimes\mbb{1}_\mc{E}/d_\mc{E}\right)$. As before, we assume that $U_t$ is such that $\rho_t=\Tr_\mc{E}(\upsilon_t)=\Phi_A(\rho)$. With that, we obtain $\fR_A^{\alpha\up}(\rho_t)=\fR_\mc{A}^{\max}=\ln{d_\mc{A}}$ and, via Axiom \ref{A_reality_inf}, 
\be\label{R_renyi_up}
\fR_A^{\alpha\up}(\rho)=\ln d-\inf_{\sigma_\mc{S}}D_\alpha\left(\upsilon_t\bigdbar\sigma_\mathcal{S}\otimes\tfrac{\mbb{1}_\mc{E}}{d_\mc{E}}\right),
\ee
where $\upsilon_t=U_t\left(\rho\otimes \ket{e_0}\bra{e_0}\right)U_t^\dag$, $d=d_\mc{A}d_\mc{E}$, and $d_\mc{E}=d_\mc{A}$. That this quantifier is indeed a reality monotone for $\alpha\in(0,1)$ (when $\fR_A^{\alpha\down}$ is too) is demonstrated in Appendix \ref{appendix_opt}. A disadvantage of $\fR_A^{\alpha\up}$ in comparison with $\fR_A^{\alpha\down}$ is the presence of the environment state $\ket{e_0}$ and the observable-dependent unitary operator $U_t$, whose formal structure is provided in the proof of Theorem \ref{theorem_uni} (Appendix \ref{appendixB}). What is more, the optimization process may be  impracticable, specially if the sandwiched R\'enyi divergence $\tD_\alpha$ is used instead of $D_\alpha$. Interestingly, however, by use of the quantum Sibson identity (see the supplemental material of Ref. \cite{sharma2013}), we obtain the closed form $\inf_{\sigma_\mc{B}}D_\alpha(\rho||\mbb{1}_\mc{A}\otimes\sigma_\mc{B})=\frac{\alpha}{\alpha-1}\ln\Tr_\mc{B}\left\{\left[\Tr_\mc{A}\left(\rho^\alpha\right) \right]^{1/\alpha} \right\}$, which allows us to simplify Eq. \eqref{R_renyi_up} as
\be\label{R_renyi_up_simp}
\fR_A^{\alpha\up}(\rho)=\ln{d_\mc{A}}-\frac{\alpha}{\alpha-1}\ln\Tr_\mc{A}\left\{\left[\Tr_\mathcal{E}\left(\upsilon_t^\alpha \right) \right]^{1/\alpha} \right\},
\ee
for $\alpha\in(0,1)$. An example is opportune.
\begin{example}
Computing the monotone \eqref{R_renyi_up_simp} for the Werner state \eqref{werner} yields the result $\fR_A^{\alpha\up}(\rho_\epsilon)=\ln 2-\frac{\alpha}{\alpha-1}\ln\chi$, where
\be 
\chi=\frac{1-\epsilon}{2}+\left[\frac{(1+\epsilon)^\alpha+(1+3\epsilon)^\alpha}{2^{\alpha+1}} \right]^{1/\alpha}
\ee 
and $\alpha\in(0,1)$. Figure \ref{fig:RRdif} illustrates the slight differences between the monotones \eqref{R_renyi_up_simp} and \eqref{R_renyi_down}, which always respect $\fR_A^{\alpha\up}(\rho_\epsilon)\geqslant\fR_A^{\alpha\down}(\rho_\epsilon)$, as expected.
\end{example}
\begin{figure}[htb]
\centering
\includegraphics[width=0.8\linewidth]{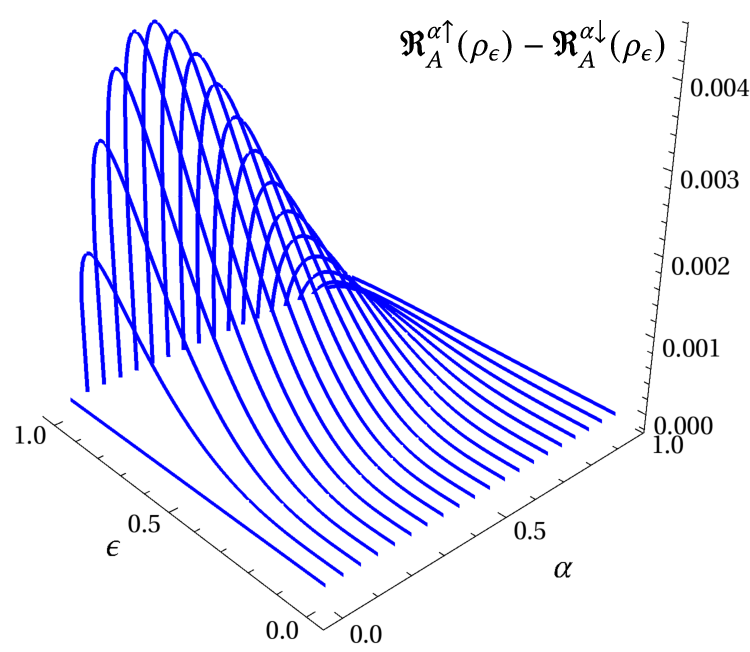}
\caption{Differences $\fR_A^{\up\alpha}(\rho_\epsilon)-\fR_A^{\down\alpha}(\rho_\epsilon)$ for the Werner state \eqref{werner} as a function of the purity parameter $\epsilon\in(0,1)$ and $\alpha\in(0,1)$. The maximum difference, $\sim 0.0044$, is reached when $(\alpha,\epsilon)\sim(0.24,0.89)$.}
\label{fig:RRdif}
\end{figure}

Roughly speaking, the divergences in Eq.~\eqref{R_renyi_up} evaluate the ``minimum distance'' (according to an ``entropic metric'') between the time evolved state $\upsilon_t$ and the product $\sigma_\mc{S}\otimes\mbb{1}_\mc{E}/d_\mc{E}$. One might argue, however, that it would be more reasonable to run the optimization, at every instant of time, within a set more closely related with the reduced state $\Tr_\mc{E}(\upsilon_t)$, which is strictly confined to the dynamics imposed by $U_t$. Adhering to this rationale, we start over by proposing the following adaptation in the optimized conditional information \eqref{Icond_Renyi_down}:
\be 
I_\mc{E|S}^{\alpha\down}(U_t\upsilon_0U_t^\dag)=\inf_{\sigma_\mc{S}}D_\alpha\left(U_t\upsilon_0U_t^\dag\bigdbar\Tr_\mc{E}\left(U_t\eta_0 U_t^{\dag}\right)\otimes\tfrac{\mbb{1}_\mc{E}}{d_\mc{E}}\right).\label{Icond_Renyi_new}
\ee 
where $\eta_0=\sigma_\mc{S}\otimes\ket{e_0}\bra{e_0}$ and $\sigma_\mc{S}\in\mf{B}(\mc{H_S})$.
For $t=0$ we see that no significant change is implied to the original definition. By use of additivity and $\upsilon_0=\rho\otimes\ket{e_0}\bra{e_0}$, we easily obtain $I_\mc{E|S}^{\alpha\down}(\upsilon_0)=\ln d_\mc{E}$. Nevertheless, for $t>0$ the optimization runs only over the initial state $\sigma_\mc{S}$. This preserves the dynamics imposed by $U_t$ and avoids any artificial freedom that would otherwise be tested throughout the minimization process. Noticing that $\Tr_\mc{E}\left(U_t\eta_0 U_t^{\dag}\right)=\Phi_A(\sigma_\mc{S})$, we can employ Theorem \ref{theorem_uni} (Appendix \ref{appendixB}), unitary invariance, and additivity to show that $I_\mc{E|S}^{\alpha\down}(\upsilon_t)=\ln{d_\mc{E}}+\inf_{\sigma_\mc{S}}D_\alpha\left(\rho||\Phi_A(\sigma_\mc{S})\right)$. Employing Axiom \ref{A_reality_inf} with $\bar{\fR}_A^{\alpha}(\rho_t)=\ln{d_\mc{A}}$ gives
\be\label{R_renyi_bar}
\bar{\fR}_A^{\alpha}(\rho)=\ln{d_\mc{A}}-\inf_{\sigma_\mc{S}}D_\alpha\left(\rho||\Phi_A(\sigma_\mc{S})\right).
\ee
This is exactly the result we would obtain by restricting the optimization in Eq. \eqref{R_renyi_up} to the set of $A$-reality states, that is, the one constituted by states satisfying $\sigma_\mc{S}=\Phi_A(\sigma_\mc{S})$.

From the discussion conducted up until now, one can conclude that
$\fR_A^{\alpha\down}(\rho)\leqslant\bar{\fR}_A^{\alpha}(\rho)\leqslant\fR_A^{\alpha\up}(\rho)$. This is not to say, however, that we have mathematical evidence that $\bar{\fR}_A^\alpha$ and $\fR_A^{\alpha\down}$ are distinct quantities. On the contrary, we do have evidence that they are equal when $\alpha\to 1$. To show this, we use Lemma \ref{lemma_f} to obtain $D\left(\rho||\Phi_A(\sigma_\mc{S})\right)=D\left(\rho||\Phi_A(\rho)\right)+D\left(\Phi_A(\rho)||\Phi_A(\sigma_\mc{S})\right)$, which gives $\inf_{\sigma_\mc{S}}D\left(\rho||\Phi_A(\sigma_\mc{S})\right)=D\left(\rho||\Phi_A(\rho)\right)$. This readily implies that, $\bar{\fR}_A^\alpha=\fR_A^{\alpha\down}$ as $\alpha\to 1$. Although we expect for the definitive solution to this problem, we can safely announce the R\'enyi reality monotone $\fR_A^\alpha(\rho)$ in the form
\be \label{R_reyni}
\fR_A^\alpha\in\{\fR_A^{\alpha\down},\bar{\fR}_A^\alpha,\fR_A^{\alpha\up} \}
\ee 
for $\alpha\in(0,1)$, with their respective formulas \eqref{R_renyi_down}, \eqref{R_renyi_bar}, and \eqref{R_renyi_up}, and $\fR_A^{\alpha\down}$ thus being a lower bound for the R\'enyi $A$-reality.

\begin{remark}\label{Rminmax}
Very similar arguments can be made toward the establishment of  reality quantifiers such as $\fR_A^{\min}$, $\tfR_A$, and $\fR_A^{\max}$. Operationally, they can be directly obtained through the replacement of $D_\alpha$ in Eq.~\eqref{R_renyi_down} by the respective divergences $D_{\min}\coloneqq \lim_{\alpha\to 0}D_\alpha$, $\tD_\alpha$ (the sandwiched R\'enyi divergence), and $D_{\max}\coloneqq \lim_{\alpha\to +\infty}\tD_\alpha$. The reader can find in Table \ref{tab:properties} a summary of the properties that are satisfied by these divergences and, in Table \ref{tab:axioms}, the axioms respected by the corresponding reality quantifiers. It turns out, though, that only $\tfR_A$ works as a reality monotone for some values of $\alpha$.
\end{remark}

\subsection{Tsallis reality monotones}

Unlike the R\'enyi divergence \eqref{renyi}, the Tsallis relative entropy \eqref{tsallis} is not additive on its entries. Yet, we show now that it is still possible to construct a reality monotone in this information theory. Aiming at accounting for Axiom \ref{A_reality_inf}, we employ Theorem \ref{theorem_uni}, unitary invariance, and definitions \eqref{tsallis} and \eqref{Icond_Tsallis} to demonstrate that $I_\mc{E|S}^q(\upsilon_t)-I_\mc{E|S}^q(\upsilon_0)=D_q\left(\rho||\Phi_A(\rho)\right)$, where $\upsilon_0=\rho\otimes\ket{e_0}\bra{e_0}$. Now, we note that for $\rho=\ket{\psi}\bra{\psi}\otimes\mbb{1}_\mc{B}/d_\mc{B}$ and $\Phi_A(\psi)=\mbb{1}/d_\mc{A}$ we find $D_q(\rho||\Phi_A(\rho))=d_\mc{A}^{q-1}S_q(\mbb{1}_\mc{A}/d_\mc{A})$, which also manifests the normalization issue. Selecting a unitary evolution such that $\fR_A^q(\rho_t)=\fR_A^q\left(\Phi_A(\rho)\right)=\ln_q{d_\mc{A}}$, we then propose the quantifier
\be\label{R_tsallis} \fR_A^q(\rho)=\ln_q{d_\mc{A}}-d_\mc{A}^{1-q}D_q\left(\rho||\Phi_A(\rho)\right),
\ee
which reduces to its von Neumann counterpart \eqref{R_vonNeumann} as $q\to 1$. Table \ref{tab:properties}, along with the fact that $D_q(\rho\otimes\Omega||\sigma\otimes\Omega)=D_q(\rho||\sigma)$, shows that likewise the R\'enyi divergences, $D_q$ satisfies all properties necessary for one to validate $\fR_A^q$ as a reality monotone in the domain $q\in(0,2]$. On the other hand, even if flagging comes to eventually be proved for the Tsallis reality monotone, the lack of additivity already guarantees that $\fR_A^q$ will never be classified as a reality measure. See Table \ref{tab:axioms} for a summary of the axioms held by the Tsallis reality monotone \eqref{R_tsallis} with respect to the parameter $q$. Next, we present a brief case study.

\begin{example}
Let us take again the Werner state \eqref{werner}. A lengthy but direct calculation of \eqref{R_tsallis} yields
\be
\fR_A^q(\rho_\epsilon)=\ln_q{2}-\frac{(1-\epsilon)^q-2(1+\epsilon)^q+(1+3\epsilon)^q}{4(q-1)\,[2(1+\epsilon)]^{q-1}},
\ee
for $q\in(0,1)\cup(1,2]$. As in Example \ref{ex_werner}, due to the rotational invariance of the singlet state it follows that $\fR_A^q(\rho_\epsilon)$ actually is observable independent. See Fig. \ref{fig:RTwerner} for numerical illustrations of the above formula. It can be checked that $\partial_q\fR_A^q(\rho_\epsilon)\leqslant 0$, meaning that $\fR_A^q(\rho_\epsilon)\geqslant \fR_A^p(\rho_\epsilon)$ for $q\leqslant p$.
\end{example}
\begin{figure}[htb]
\centering
\includegraphics[width=0.9\linewidth]{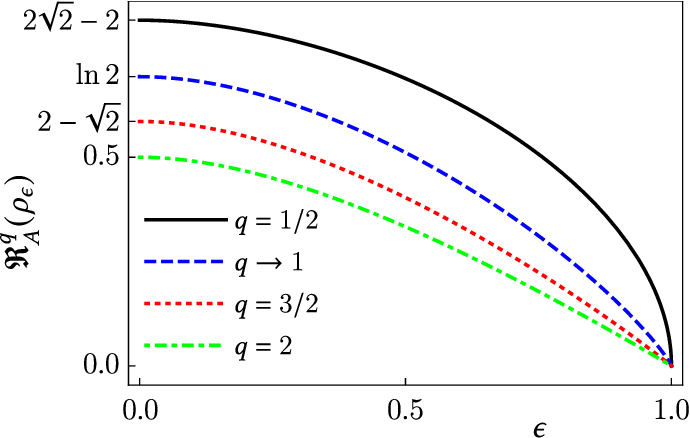}
\caption{Tsallis reality monotone $\fR_A^q(\rho_\epsilon)$ for any spin observable $A$ of the first qubit of a Werner state \eqref{werner} as a function of the purity parameter $\epsilon$ (as introduced in Example \ref{ex_werner}) for: $q=1/2$ (solid black line), $q\to 1$ (dashed blue line), $q=3/2$ (dotted red line) and $q=2$ (dash-dotted green line).}
\label{fig:RTwerner}
\end{figure}

\newcolumntype{A}{l}
\newcolumntype{B}{>{\centering}p{0.085\textwidth}}
\newcolumntype{C}{c}
\begin{table*}[ht]
\renewcommand{\arraystretch}{1.5}
\centering
\begin{tabular}{ABCBCBC}
\hline
& $\fR_A$ & $\fR_A^{\alpha}$ & $\fR_A^{\min}$ &  $\tfR_A^{\alpha}$ & $\fR_A^{\max}$ & $\fR_A^q$ \\
\hline
Axiom \ref{A_reality_inf} (information flow) & \cmark & \cmark & \cmark & \cmark & \cmark & \cmark \\
Axiom \ref{A_reality_meas} (measurements) & \cmark & \cmark & \xmark & \cmark & \cmark & \cmark \\
Axiom \ref{A_role_parts}(a)  (part discard) & \cmark & $\alpha\in(0,1)\cup(1,2]$ & \cmark & $\alpha\in[1/2,1)\cup(1,+\infty)$ & \cmark & $q\in(0,1)\cup(1,2]$\\
Axiom \ref{A_role_parts}(b)  (uncorrelated part) & \cmark & \cmark & \cmark & \cmark & \cmark & \cmark\\
Axiom \ref{A_uncertainty} (uncertainty relation) & \cmark & \cmark & \xmark & \cmark & \cmark & \cmark\\
Axiom \ref{A_mix} (mixing) & \cmark & $\alpha\in(0,1)$ & \cmark & $\alpha\in[1/2,1)$ & \xmark & $q\in(0,1)\cup(1,2]$ \\
\hline
Axiom \ref{A_add} (additivity) & \cmark & \cmark & \cmark & \cmark & \cmark & \xmark\\
Axiom \ref{A_flag} (flagging) & \cmark & ? & ? & ? & \xmark & ?\\
\hline
\end{tabular}

\caption{Summary of the axioms satisfied by the $A$-reality quantifiers $\fR_A$ [Eq.~\eqref{R_vonNeumann}], $\fR_A^{\alpha}$ [Eq.~\eqref{R_reyni}], $\fR_A^{\min}$, $\tfR_A^{\alpha}$, $\fR_A^{\max}$ (see Remark \ref{Rminmax}), and $\fR_A^q$ [Eq.~\eqref{R_tsallis}] built out of their corresponding divergences $D$, $D_\alpha$, $D_{\min}$, $\tD_\alpha$, $D_{\max}$, and $D_q$, whose properties are listed in Table \ref{tab:properties}. For some specific parameter domains, our approach legitimates several reality monotones, namely, $\fR_A^{\alpha}$, $\tfR_A^{\alpha}$, and $\fR_A^q$, while only $\fR_A$ can be validated (up to now) as a reality measure.}

\label{tab:axioms}

\end{table*}

\section{Concluding remarks}\label{secDiscussion}

Inspired by a significant amount of theoretical and experimental works regarding the emergence of classicality from the quantum substratum \cite{dieguez2018,zurek2009,ciampini2018,chen2019,unden2019,bilobran2015,gomes2018,orthey2019,fucci2019,engelbert2020,mancino2018,dieguez2021,costa2020}, here we propose an axiomatization for the concept of quantum realism. This notion is different from classical reality in a very fundamental manner, namely, noncommuting observables cannot be simultaneous elements of reality in general (Axiom \ref{A_uncertainty}). Our core premise, implemented via Axiom \ref{A_reality_inf}, is the one permeating the aforementioned literature: an observable $A$ emerges as an element of the physical reality only when another degree of freedom encodes information about it. By its turn, Axiom \ref{A_reality_meas} highlights the role of measurements to quantum realism. In full consonance with Axiom \ref{A_reality_inf}---for measurements can be viewed as dynamical processes through which an apparatus get to know about the measured observable---the second axiom can yet be used, along with the Stinespring theorem [Eq. \eqref{stinespring}], as a necessary criterion for one to decide when a measurement is concluded. The rationale is that we do not expect a measurement to have been finished before the establishment of reality, that is, before the instant $t$ at which $\rho_t=\Phi_A(\rho)$ and hence $\fR_A(\rho_t)=\fR_\mc{A}^{\max}$. The role of large environments in this respect then consists of ensuring the irreversibility of the measurement. Axioms \ref{A_role_parts} and \ref{A_mix} complete the list of reasonable assumptions for a functional $\fR_A(\rho)$ to be named a reality monotone, while Axioms \ref{A_add} and \ref{A_flag} are additional conditions legitimating a reality measure. However debatable our list of axioms may be, it furnishes an intuitive ``metric independent'' characterization of quantum realism, thus framing the concept in a formal structure. Moreover, as we have explicitly demonstrated (see Table \ref{tab:axioms} for an overview), sensible reality monotones and a reality measure can be built by use of information theoretic quantities associated with the von Neumann, R\'enyi, and Tsallis entropies. 

At least two technical questions are left open for future research. The first one concerns the completion of the last line of Table \ref{tab:axioms}. Indeed, the concept of flagging has been introduced only recently and formal results in this regard are still lacking for the R\'enyi and Tsallis divergences. Second, it would be useful, mainly for operational purposes, to have a picture of whether ``metrics'' other than the entropic ones, such as norm-based quantifiers, can be used as sensible reality monotones. Finally, with regard to resource theories, although some evidence has been put forward suggesting that the $A$-irreality, $\fI_A(\rho)=\ln{d_\mc{A}}-\fR_A(\rho)$, can be viewed as a quantum resource \cite{costa2020}, it would be interesting to have at hand a concrete information task wherein this concept configures a clear  advantage in relation to contexts involving the $A$-reality state $\Phi_A(\rho)$. 

\begin{acknowledgments}
This research was financed in part by the Coordena\c{c}\~ao de Aperfei\c{c}oamento de Pessoal de N\'ivel Superior - Brasil (CAPES) - Finance Code 001. A.C.O. thanks Valber S. Gomes, Ana C. S. Costa, and Luis G. Longen for useful discussions. R.M.A. acknowledges support from CNPq/Brazil (Grant No. 309373/2020-4) and the National Institute for Science and Technology of Quantum Information (INCT-IQ/CNPq, Brazil).
\end{acknowledgments}

\appendix

\section{Special cases of the R\'enyi divergences}
\label{appendix_cases}

The \textit{min-relative entropy}, as defined by Datta \cite{datta2009}, is the limit of the original R\'enyi divergence as $\alpha\to 0$  . Operationally, it comes as follows. Consider the spectral decompositions   $\rho=\sum_i r_i\ket{\lambda_i}\bra{\lambda_i}$ and $\sigma=\sum_j s_j\ket{\nu_j}\bra{\nu_j}$. Since $f(\rho)=\sum_i f(r_i)\ket{\lambda_i}\bra{\lambda_i}$, then we have $\rho^0=\sum_{i:\, r_i>0} \ket{\lambda_i}\bra{\lambda_i}$, which is called the \textit{projection onto the support of $\rho$}. Thus, the min-relative entropy is given by $D_0(\rho||\sigma)=-\ln\Tr\rho^0\sigma$. Explicitly, we have
\be
D_{\alpha\to 0}(\rho||\sigma)=-\ln\sum_j\sum_{i:\, r_i>0}s_j|\braket{\lambda_i|\nu_j}|^2.
\ee
Note that, the min-relative entropy does not satisfy continuity nor positive definiteness in $\rho$ and $\sigma$, a feature that prevents it to satisfy Axiom \ref{A_reality_meas}. One special case of the sandwiched R\'enyi divergence is the \textit{collision relative entropy}, which was introduced in Ref. \cite{renner2005} in its conditional form as a generalization of the classical conditional collision entropy to the quantum theory. It is obtained when we choose $\alpha=2$ in \eqref{sandwiched}:
\be
\tD_2(\rho||\sigma)=\ln\Tr\left[\left(\sigma^{-\frac{1}{4}}\rho\sigma^{-\frac{1}{4}} \right)^2\right].
\ee
Another special case of $\tD_\alpha$ is obtained as $\alpha\to +\infty$, which is called \textit{max-relative entropy} \cite{datta2009}:
\be
\tD_{\alpha\to+\infty}(\rho||\sigma)=\ln\bignorm{\sigma^{-\frac{1}{2}}\rho\sigma^{-\frac{1}{2}}}{\infty}.
\ee
Here, the \textit{operator norm} $\norm{\varrho}{\infty}$ is given by the maximum eigenvalue of a density state $\varrho$.

\section{Sudsidiary results}\label{appendixB}

The results presented in this section refers to states such that $\rho\in\mf{B}(\mc{H_S})$, with $\mc{H_S=H_A\otimes H_B}$, and the unrevealed measurements map \eqref{map_phi}.
\begin{theorem}\label{theorem_uni}
	Let the unitary evolution $U_t$ be defined by the Stinespring dilation theorem \eqref{stinespring} with $\epsilon=1$. It follows that $U_t$ commutes with $\Phi_A(\rho)\otimes\mbb{1}_\mc{E}/d_\mc{E}$, that is,
	\be
	U_t \left(\Phi_A(\rho)\otimes\tfrac{\mbb{1}_\mc{E}}{d_\mc{E}}\right) U_t^\dag=\Phi_A(\rho)\otimes\tfrac{\mbb{1}_\mc{E}}{d_\mc{E}}.
	\ee
\end{theorem}
\begin{proof}
Take the joint state $\upsilon_0=\rho\otimes\ket{e_0}\bra{e_0}\in\mathfrak{B}(\mathcal{H_S\otimes H_E})$, with $d_\mathcal{E}=\dim\mathcal{H_E}=\dim\mathcal{H_A}$. Write the unitary operator
\be\label{unitary}
U_t=\sum_{k=0}^{d_\mc{E}-1} P_k\otimes T_k,
\ee
where $P_k=A_k\otimes\mbb{1}_\mc{B}$ is a subspace projector and $T_k$ is a unitary operator satisfying $T_kT_k^\dag=T_k^\dag T_k=\mbb{1}_\mc{E}$ and $\braket{e_0|T_j^\dag T_i|e_0}=\delta_{ij}$.  An example of this structure is provided by the shift operator $T_k\ket{e_i}=\ket{e_{i+k}}$ in the cyclic space with the boundary condition $T_1\ket{e_{d_\mathcal{E}-1}}=\ket{e_{d_\mathcal{E}}}=\ket{e_0}$. Its matrix representation is given by a power of the generalized Pauli operator $\sigma_x$,
\be\label{gen_pauli_x}
T_k\stackrel{\cdot}{\equiv} \begin{pmatrix}
0 & 0 & 0 & \cdots & 0 & 1 \\
1 & 0 & 0 & \cdots & 0 & 0 \\
0 & 1 & 0 & \cdots & 0 & 0 \\
0 & 0 & 1 & \cdots & 0 & 0 \\
\vdots & \vdots & \vdots & \ddots & \vdots & \vdots \\
0 & 0 & 0 & \cdots & 1 & 0
\end{pmatrix}^k.
\ee
Notice that $T_iT_j=T_{i+j}$, which renders $U_tU_t=\sum_{i}P_i\otimes T_{2i}$.
One has $\upsilon_t=U_t\upsilon_0U_t^\dag=\sum_{i,j}P_i\rho P_j\otimes T_i\ket{e_0}\bra{e_0}T_j^\dag$, which correctly reproduces the Stinespring relation
\be
\Tr_\mathcal{E}(\upsilon_t)=\sum_{ij}P_i\rho P_j \braket{e_0|T_j^\dag T_i|e_0}
=\Phi_A(\rho),
\ee
The unitary operator \eqref{unitary} is unique up to a unitary operation over the environment. Finally, by direct application of $U_t$ we have $U_t\left(\Phi_A(\rho)\otimes\frac{\mbb{1_\mc{E}}}{d_\mc{E}} \right)U_t^\dagger=\sum_{kij} P_iP_k\rho P_k P_j \otimes T_i (\mbb{1}_\mc{E}/d_\mathcal{E}) T_j^\dagger$, which yields the desired result since $P_i$ is a projector and  $T_i$ is unitary.
\end{proof}
\noindent As a corollary we have  $\Phi_A\left(\Phi_A(\rho)\right)=\Phi_A(\rho)$. This shows that once a state of reality is established by the conjugation of a unitary evolution and a discard, then repeating this operation is innocuous, as maximum reality cannot be further enhanced.

\begin{lemma}\label{lemma_f}
For any function $f$, any quantum state $\rho$, and any observable $A$, one has $\Tr[\rho\, f\left(\Phi_A(\rho)\right)]=\Tr[\Phi_A(\rho)f\left(\Phi_A(\rho)\right)]$.
\end{lemma}
\begin{proof}
\noindent See the Appendix of Ref. \cite{costa2013}.
\end{proof}

\begin{lemma}\label{lemma_bound}
Given the reality state $\Phi_A(\rho)=\sum_ip_iA_i\otimes\rho_{\mc{B}|i}$, it holds that $D\left(\rho||\Phi_A(\rho)\right)\leqslant S\left(\Phi_A(\rho_\mc{A})\right)\leqslant \ln{d_\mc{A}}$.
\end{lemma}
\begin{proof}
From Lemma \ref{lemma_f} one can straightforwardly show that $D\left(\rho||\Phi_A(\rho)\right)=S(\Phi_A(\rho))-S(\rho)$. Since $S(\rho)\geqslant\sum_ip_iS(\rho_{\mc{B}|i})$ (see Lemma 2 of Ref. \cite{xi2011}), we can employ the joint entropy theorem $S\left(\Phi_A(\rho)\right)=H(\{p_i\})+\sum_ip_iS(\rho_{\mc{B}|i})$ \cite{nielsen2000}, with $H(\{p_i\})$ being the Shannon entropy of the distribution $p_i$, to finally obtain $D\left(\rho||\Phi_A(\rho)\right)\leqslant H(\{p_i\})=S\left(\Phi_A(\rho_\mc{A})\right)\leqslant \ln{d_\mc{A}}$.
\end{proof}

\begin{lemma}\label{lemma_uncertainty}
Consider generic observables $X,Y\in\mf{B}(\mc{H_A})$ and the von Neumann reality quantifier \eqref{R_vonNeumann}. It follows that $\fR_X(\rho)+\fR_Y(\rho)\leqslant 2\ln{d_\mc{A}}$, with equality iff $\rho=\Phi_X(\rho)=\Phi_Y(\rho)$.
\end{lemma}
\begin{proof}
By Eqs. \eqref{D>0} and \eqref{R_vonNeumann} we see that the main claim is readily satisfied and that the equality holds iff $\rho=\Phi_X(\rho)=\Phi_Y(\rho)$, meaning that $\rho$ must be a state of simultaneous reality for $X$ and $Y$. This will certainly be the case when $[X,Y]=0$, for $X$ and $Y$ will share the same set of eigenstates so that $\Phi_X=\Phi_Y$, but also for $\rho=(\mbb{1}_\mc{A}/d_\mc{A})\otimes\rho_\mc{B}$, as can be checked by direct calculation.
\end{proof}
\noindent Note that this proof is valid also for a reality quantifier that is based on any divergence measure that respects the positive definiteness property [Eq.~\eqref{D>0}].

\begin{lemma}\label{lemma_delta}
Consider observables $X,Y\in\mathfrak{B}(\mathcal{H_A})$ and the von Neumann reality quantifier \eqref{R_vonNeumann}. If $\Phi_{XY}=\Phi_{YX}$, then the monitoring of $Y$ never decreases the reality of $X$, that is,
\be\label{delta_RXY}
\Delta\coloneqq\fR_X(\mc{M}_Y^\epsilon(\rho))-\fR_X(\rho)\geqslant 0,\qquad\forall\epsilon\in[0,1].
\ee
\end{lemma} 
\begin{proof}
In light of Axiom \ref{A_mix} (mixing) and definition \eqref{monitoring}, we find $\Delta\geqslant \epsilon\left[\fR_X(\Phi_Y(\rho))-\fR_X(\rho) \right]$, which can be explicitly expressed as $\Delta\geqslant\epsilon\left[D\left(\rho||\Phi_X(\rho)\right)-D\left(\Phi_Y(\rho)||\Phi_{XY}(\rho)\right) \right]$, where $\Phi_{XY}(\rho)\equiv\Phi_X\Phi_Y(\rho)$. Using the hypothesis and DPI we can write $D\left(\Phi_Y(\rho)||\Phi_{XY}(\rho)\right)=D\left(\Phi_Y(\rho)||\Phi_{YX}(\rho)\right)\leqslant D\left(\rho||\Phi_X(\rho)\right)$, which proves that $\Delta\geqslant 0$ $\forall\,\epsilon\in[0,1]$, as desired. It is worth noticing that, apart from the trivial scenario where $[X,Y]=0$, the hypothesis is true also when $X$ and $Y$ are maximally noncommuting, that is, when their eigenstates form MUBs satisfying $|\braket{x_i|y_j}|=1/\sqrt{d_\mc{A}}$.
\end{proof}

\vskip0.01mm
\section{Proof that $\fR_A^{\alpha\up}$ is a reality monotone}\label{appendix_opt}

By construction, $\fR_A^{\alpha\up}$ is in harmony with Axiom \ref{A_reality_inf}. We now prove that $\fR_A^{\alpha\up}(\rho)=\fR^{\max}$ iff $\rho=\Phi_A(\rho)$, as per Axiom \ref{A_reality_meas}. The claim is true iff $D_\alpha\left(\upsilon_t||\bar{\sigma}_\mc{S}\otimes\mbb{1}_\mc{E}/d_\mc{E}\right)=\ln{d_\mc{E}}$, where $\bar{\sigma}_\mc{S}$ is the solution for the minimization and $\upsilon_t=U_t\left(\rho\otimes \ket{e_0}\bra{e_0} \right)U_t^\dag$. Choosing $\bar{\sigma}_\mc{S}=\Phi_A(\bar{\sigma}_\mc{S})$ (an $A$-reality state), one can apply Theorem \ref{theorem_uni} and unitary invariance to obtain $D_\alpha\left(\rho||\Phi_A(\bar{\sigma}_\mc{S})\right)=0$, which can be reached iff $\rho=\Phi_\mc{A}(\bar{\sigma}_\mc{S})$, meaning that $\rho$ is an $A$-reality state satisfying $\Phi_A(\rho)=\rho$. Axioms \ref{A_reality_meas} and \ref{A_role_parts}(a) are satisfied directly from DPI and the fact that $U_t^{\epsilon}$ in Eq. \eqref{stinespring} commutes with $\mc{M}_A^\epsilon$ and $\Tr_{\mathcal{X}}$ for $\mc{H_X}\subseteq\mathcal{H_B}$. Provided the optimization is made over $\sigma_\mc{S}\otimes\Omega\in\mathfrak{B}(\mathcal{H_S\otimes H}_\Omega)$, Axiom \ref{A_role_parts}(b) is trivially satisfied via the additivity property. Also, because Axiom \ref{A_reality_meas} applies, we can use the arguments employed for the proof of Lemma \ref{lemma_uncertainty} to show that Axiom \ref{A_uncertainty} is also true for $\fR_A^{\alpha\up}$. Finally, the validity of Axiom \ref{A_mix} is immediately verified by the convexity of the conditional information \eqref{Icond_Renyi_down}. Although additivity guarantees the agreement with Axiom \ref{A_add}, the flagging property has not yet been demonstrated for the quantity \eqref{Icond_Renyi_down}, which precludes $\fR_A^{\alpha\up}$ to be promoted to the status of a reality measure for $\alpha\in(0,1)$.


\end{document}